%% file: main.tex
\documentclass{article}
\usepackage[utf8]{inputenc}

\newcommand{\tup}[1]{\mathbf{#1}}
\newcommand{\eat}[1]{}

\usepackage{amsmath, amssymb, amsfonts}

\usepackage{microtype}
\usepackage{multirow}

\usepackage{tikz}
\usepackage{tikz-qtree}
\usetikzlibrary{shapes.geometric}
\usepackage{subcaption}
\usepackage{url}

\newtheorem{theorem}{Theorem}[section]

\newtheorem{proposition}[theorem]{Proposition}

\newtheorem{example}[theorem]{Example}

\newtheorem{definition}[theorem]{Definition}
\newtheorem{proof}[theorem]{Proof}

\usepackage{paralist}

\title{Equivalence-Invariant Algebraic Provenance for Hyperplane Update Queries}
\author{
  Pierre Bourhis\\
   \small{NRS, UMR 9189}\\ \small{CRIStAL}
  \and
  Daniel Deutch\\
   \small{Tel Aviv    University}
  \and
  Yuval Moskovitch\\
   \small{Tel Aviv    University}
}
\date{}

\begin{document}

\maketitle

\begin{abstract}
    The algebraic approach for provenance tracking, originating in the semiring model of Green et. al, has proven useful as an abstract way of handling metadata.
    Commutative Semirings were shown to be the ``correct" algebraic structure for Union of Conjunctive Queries, in the sense that its use allows provenance to be invariant under certain expected query equivalence axioms.

 In this paper we present the first (to our knowledge) algebraic provenance model, for a fragment of update queries,
 that is invariant under set equivalence. The fragment that we focus on is that of hyperplane queries, previously studied in multiple lines of work.
 Our algebraic provenance structure and corresponding provenance-aware semantics are based on the sound and complete axiomatization of Karabeg and Vianu.
 We demonstrate that our construction can guide the design of concrete provenance model instances for different applications. We further study the
 efficient generation and storage of provenance for hyperplane update queries.
 We show that a naive algorithm can lead to an exponentially large provenance expression,
 but remedy this by presenting a normal form which we show may be efficiently computed alongside
 query evaluation.
 We experimentally study the performance of our solution and demonstrate its scalability and usefulness, and in particular
 the effectiveness of our normal form representation.
\end{abstract}

\input{intro-daniel}
\input{prelim}
\input{model}

\input{provUse}

\input{provMin}

\input{exp}

\input{related}

\input{conc}

\paragraph*{Acknowledgements}
This research has been funded by the European Research Council (ERC) under the European Unions Horizon 2020 research and innovation 
programme (Grant agreement No. 804302), and the Israeli Science Foundation (ISF) Grant No. 978/17. Pierre Bourhis is supported by the ANR Project Headwork  ANR-16-CE23-0015.

\bibliographystyle{abbrv}
\bibliography{bibtex}

\end{document}

%% file: intro-daniel.tex
\vspace{-0.3cm}
\section{Introduction}


The tracking of {\em provenance} for database queries has been
extensively studied in the past years (see e.g.
\cite{CheneyProvenance, GKT-pods07, ProvenanceBuneman, pods11}). In
a nutshell, data provenance captures details of the computation that
took place and resulted in the generation of each output data item.
Multiple models for data provenance have been proposed, for multiple
query languages such as the (positive) relational algebra, datalog
(see \cite{GKT-pods07}), data-intensive workflows (e.g.,
\cite{DavidsonF08, karma}) data mining \cite{GlavicSAM13}, and
data-centric applications \cite{vldbj15}. Provenance has been proven
useful for managing access control, trust, hypothetical reasoning,
view maintenance and debugging (see
\cite{GeertsUKFC16,GreenKIT07,Caravan, GKT-pods07, icde16}).

 The approach advocated by \cite{GKT-pods07} is based on designing
 {\em algebraic} provenance structures whose equivalence axioms
 are based on equivalences in the formalism for which provenance is designed to be tracked.
 This guarantees that by design, equivalent queries/programs in the formalism of interest will
 have equivalent provenance for their output. In a sense, this means that provenance captures the
 ``essence of computation" that has been performed. The commutative semiring model of \cite{GKT-pods07} achieves this property
 for the positive relational algebra; several extensions have been studied \cite{pods11, glavic, KarvounarakisGIT13} for different query languages.

In this paper we focus on a fragment of {\em update queries} and sequences thereof (which we refer to as ``transactions"),
 and propose a novel algebraic provenance model. The fragment of update queries that we focus on
 is that of {\em hyperplane queries}, introduced in \cite{AbiteboulV88} as simple yet important building blocks of transactions. Hyperplane queries are intuitively ``domain-based", in that selection of tuples in each query only involves the inspection of individual attribute values for each tuple. As demonstrated in \cite{AbiteboulV88,vianu}, this fragment of transactions facilitates appealing theoretical features, while allowing to express transactions of interest. Specifically, for this fragment, \cite{vianu} has shown a sound and complete axiomatization, which is crucial for our provenance model as we next explain. 
 
The provenance annotations in our model are initially
 assigned to both queries and tuples; those assigned to queries are propagated to the tuples that these queries affect, so that the result of applying
 an annotated transaction is an annotated database. Then, in a similar vein to the commutative semiring model mimicking the
 equivalence axioms of positive relational algebra, our model is based on the sound and complete axiomatization
 for set equivalence of transactions in \cite{vianu}. Namely, we start with a most generic structure that uses abstract operations
 to capture the effect of each type of update query, and then introduce, for each of the axioms in \cite{vianu},
 a corresponding axiom in our algebraic structure. As we will show, this leads to a provenance framework that has the following
 favourable property: two transactions are ``provenance-equivalent", i.e., their application on every input database yields the same annotated database, if and only if they are set-equivalent.
This means that provenance in our framework is independent of the particular way that the transaction is executed and of any optimizations that may take place.  To our knowledge, ours is the first provenance model to satisfy this property for transactions (see discussion of previously proposed models in Sections \ref{sec:setEq} and \ref{sec:related}). Details of our provenance model appear in Section \ref{sec:model}.

By propagating annotations that are assigned to both queries and tuples, we are able to support multiple applications of interest, which we overview in Section \ref{sec:app}.  For instance, analysts may use the resulting provenance to conduct {\em hypothetical reasoning} with respect to both the database and transaction. Namely, by assigning  truth values to tuple and/or query annotations in the resulting provenance expressions, they may observe the effect of deleting a tuple or aborting a transaction, on the computation result. Additional examples include the support of access control, where each tuple/query is associated with access credentials and these are propagated so that we compute access credentials for each output tuple; and a ``certification" example, where we assign trust level to each tuple/query and correspondingly produce certifications to output tuples we trust. As is the case with previous algebraic provenance constructions, the idea is that we may first compute an expression in the ``most general" structure (detailed in Section \ref{sec:model}), and then upon request ``specialize" (map) it to any  application domains such as those we have just exemplified.  

Further, while our generic structure is quite complex, we
provide a ``prescription" for building instances of it. This is
achieved by establishing a connection with the commutative semiring
model: we show that
 for a simple-to-define class of commutative semirings (see Theorem \ref{th:semring} for details),
 their operators can be easily extended to define operators for our model that do satisfy the axioms.  
 
We then (Section \ref{sec:provmin}) turn to the problem of efficient provenance generation and
storage, for the ``most general" structure.  The model definition already entails an algorithm for
provenance generation, but we show that it may lead to an
exponential blowup of the provenance size with respect to the
transaction length (number of queries). We show that this blowup may
be avoided, leveraging our axioms: we derive algebraic simplification rules that are entailed by the axioms, and consequently  propose a ``normal form" structure for provenance. We show that every provenance expression obtained by applying a
sequence of hyperplane updates may be transformed to this structure.
The expressions that we obtain in this structure are far more
compact: they are in fact linear in the size of the transaction and
input database. Furthermore, we show that we can generate
expressions in this structure on-the-fly during query evaluation,
avoiding a detour through the exponentially large representation.

Finally, we present (Section \ref{sec:exp}) an experimental study of our  framework using the
TPC benchmark as well as a synthetic dataset. The experiments focus on the time
and space overheads incur by provenance tracking, and on the time it takes to ``specialize" provenance once it is computed, i.e., assign values to variables and thereby use it in applications such as described above. Our measurements are performed for implementations with and without the normal form optimization. 
Our results show that our rewrite of provenance
into its normal form (made possible due to our axioms) significantly reduces the provenance size, and may be efficiently performed alongside with provenance generation. Thereby, it also significantly benefits provenance applications, accelerating provenance use (assignment of values). 




%% file: prelim.tex
    \vspace{-0.2cm}
    \section{Preliminaries}
\label{sec:prelim}
Our goal is to define an algebraic provenance model for updates.
In this paper, we focus on the class of ``domain-based''  updates defined in \cite{AbiteboulV88}.
This class is a standard model that was studied e.g., in
\cite{MontesiT02, MontesiT95, vianu}. Importantly, \cite{vianu} has
proposed a sound and complete axiomatization for this fragment,
which will serve as a basis for our algebraic provenance model. We
describe the class of ``domain-based'' transactions in a
datalog-like language, similar to the one in \cite{icde16}.

{\bf Relational Databases.} A \emph{relational schema} is defined
over a set of relational names. A \emph{relation} has a relation
name $R$ and a set of attributes
denoted by $\mathrm{att}(R)$ . Let $\mathcal{V}$ be an infinite set
of values. A tuple $t$ of relation $R$ is a function associating
with each attribute of $R$, a value of $\mathcal{V}$. An
\emph{instance} $I$ of a relation $R$ is a set of tuples. A database
$D$ of a relational schema associates with each relation name $R$ in
the schema an instance, denoted by $R(D)$.


{\bf Hyperplane Update queries.} We next recall the definition of
update queries from \cite{AbiteboulV88} for the class of
``domain-based" transactions, where the selection of tuples only
involves the inspection of individual attribute values for each
tuple. 
We restrict the updates queries of \cite{icde16} to those
equivalent to a member of this class.

To this end, we use the notation  $R(\tup{u})$ where $\tup{u}$ is a tuple with the same arity as $R$, that may contain constants and variables. A variable $A$ in $\tup{u}$ may further be associated with a disequality expression $[A \neq a]$, restricting assignments so that the attribute in the corresponding position may not be assigned the value $a$. We say that a tuple $t \in R$ satisfies $\tup{u}$ and write $t \vDash \tup{u}$ if $t$ corresponds to an instantiation of the variables of $\tup{u}$ that satisfy the conditions.


\begin{example}
    Figure \ref{fig:initTable} shows a fragment of a {\em products} table
    in
    an E-commerce application. It includes information about the products in stock,
    their categories and price (ignore the annotations next to tuples for now).
    The following is an hyperplane query used to describe all products in the Sport category except for the ``Kids mountain bike":
$$products([p\neq ``\text{Kids mnt bike}"],\text{``Sport"}, c)\text{:-}$$
The tuple $products(\text{``Tennis Racket"},\text{``Sport"}, \$70)$
satisfies the conditions specified in the query.
    \end{example}

%

\begin{figure*}
    \centering
\tiny{
    \begin{minipage}[b]{0.45\linewidth}
            \centering
        \begin{tabular}{llcl}
            \cline{1-3}
            Product     & Category & Price &  \\ \cline{1-3}
            Kids mnt bike    & Sport & \$120 & $p_1$    \\ 
            Tennis Racket & Sport & \$70 & $p_2$    \\ 
            Kids mnt bike    & Kids & \$120 & $p_3$    \\ 
            Children sneakers& Fashion & \$40 & $p_4$    \\ \cline{1-3}
        \end{tabular}
        \subcaption{Initial Table}
        \label{fig:initTable}

    \end{minipage}%
    \begin{minipage}[b]{0.45\linewidth}
            \centering
        \begin{tabular}{llcl}
            \cline{1-3}
            Product     & Category & Price &  \\ \cline{1-3}
            Kids mnt bike    & Bicycles & \$120 & $(p_1+ p_3)\cdot_M p$    \\ 
            Tennis Racket & Sport & \$70 & $p_2$    \\ 
            Lego bricks & Kids & \$90 & $p$    \\ \cline{1-3}
            \\
        \end{tabular}
        \subcaption{Updated Table}
  
        \label{fig:updatedTable}
    \end{minipage}
}

    \caption{Products Table}\label{fig:database}

\end{figure*}


\textbf{Insertion.}\ An insertion query $Q$  is an expression
$R^{+}(\tup{u})$:- where $\tup{u}$ is a tuple of constants with the
same arity as $R$. The effect of $Q$ applied to a database $D$,
denoted by $Q(D)$, is the insertion of  $\tup{u}$ to $R(D)$.

\begin{example}\label{ex:insertQuries}

The query
        $$Products^{+}(\text{``Lego bricks", ``Kids", }\$90)\text{:-}$$
        is an example of an insertion query, adding the tuple (``Lego bricks", ``Kids", \$90) to the $Products$ table.
\end{example}
Note that each insertion query inserts a single tuple, as in~\cite{vianu}; we will consider transactions as means for inserting a
bulk of tuples.


\textbf{Deletion.}\ A deletion query $Q$ is an expression
$R^{-}(\tup{u})$:-,  where $\tup{u}$ is a tuple with the same arity
as $R$, that may contain constants and variables, possibly
associated with disequalities.  $Q(D)$ is the resulting database
obtained from $D$ by deleting all tuples of its relation $R$ that
satisfy $\tup{u}$.


%

\begin{example}\label{ex:deleteQuries}
    Reconsider the database fragment presented in Figure \ref{fig:initTable}. The query
    $$Products^{-}(a, \text{ ``Fashion", }b)\text{:-}$$
    deletes all tuples in the fashion category.
\end{example}

\textbf{Modification.}\ A modification query $Q$ is an expression
$R^{M}(\tup{u_1, u_2})\text{:- }$, where $\tup{u_1} = (u^1_0,
\ldots, u^1_n)$ and $\tup{u_2} = (u^2_0, \ldots, u^2_n)$  have the
same arity as $R$ and may contain variables and constants such that
either $u_i^1 = u_i^2$ (and then the value for this attribute
remains intact)  or $u_i^2$ is a constant (and then the value is
changed to  $u_i^2$). I.e. the constants present in $\tup{u_2}$
which are different from the corresponding variables/constants in
$\tup{u_1}$ indicate how instantiations of $\tup{u_1}$ are modified.
The result of applying $Q$ to a database $D$ is defined as follows:
for each valid assignment to $\tup{u_1}$ and $\tup{u_2}$, the tuple
$t$ of $R$ whose values correspond to the instantiation of
$\tup{u_1}$ is deleted; the tuple $t'$ whose values correspond to
the instantiation of $\tup{u_2}$ is inserted. We use $t
\rightsquigarrow t'$ to denote that $t$ was updated to $t'$.

\begin{example}\label{ex:modifyQuries}
    The query\\
   $Products^{M}(\text{``Kids mnt bike"}, a, b, \text{``Kids mnt bike", ``Bicycles"}, b)\text{:-}$

    is a modification query. Applying the query to the database fragment
    shown in Figure \ref{fig:initTable} results in an update of the category (second attribute)
    of the product ``Kids mnt bike" to ``Bicycles".
    Namely, we have that (``Kids mnt bike", ``Sport", \$120)$ \rightsquigarrow $ (``Kids mnt bike", ``Bicycles", \$120) and (``Kids mnt bike", ``Kids", \$120)$ \rightsquigarrow $ (``Kids mnt bike", ``Bicycles", \$120).

\end{example}

A \emph{transaction} $T$ is a sequence of update queries. Its semantics
with respect to a given database $D$ is that the update queries are
applied sequentially, with each query in the sequence being applied
to the result of the transaction prefix that preceded it. The database instance resulting from the application of $T$ over $D$ is denoted by $T(D)$.

The result of applying the update queries from Examples
\ref{ex:insertQuries}, \ref{ex:deleteQuries} and
\ref{ex:modifyQuries} as a sequence to our example relation, is
shown in Figure \ref{fig:updatedTable}.

\paragraph*{Note.} Hyperplane queries correspond to the following fragment of SQL: (1) tuple insertions; (2) deletion using statements of the following form: \textsc{DELETE FROM} \emph{RelationName} \textsc{WHERE} $s_1, \cdots s_m$, in which each $s_i$ is of the form \emph{AttributeName} op $c$, where op  is in $\{=,\neq\}$ and $c$ is a constant value; (3) updates using statements of the form: \textsc{UPDATE} \emph{RelationName} \textsc{SET} $l_1, \cdots, l_n$ \textsc{WHERE} $s_1, \cdots s_m$, in which each $l_i$ and $s_i$ is of the form \emph{AttributeName} op $c$, where op  is in $\{=,\neq\}$ and $c$ is a constant value. This fragment has been identified in \cite{vianu} and subsequent works as an important building block of transactions, even though it does not capture the full generality of SQL. For instance, hyperplane queries cannot capture comparison between values inside the same tuple, or subqueries  in the \textsc{WHERE} condition.

%% file: model.tex
\section{Provenance Model}
\label{sec:model}

We define an algebraic provenance model for transactions whose
design follows the following principle: introduce the most general
model that is still insensitive to rewriting under (set)
equivalence. This is in line with the approach advocated for in \cite{GKT-pods07}: the main idea is to start by having a domain of basic annotations (which one may consider as identifiers), and to define the effect of query operators over these annotations via generic algebraic operations. The resulting provenance is then a symbolic algebraic expressions over basic annotations. The next step is to  
add equivalence axioms to the structure so that semantically equivalent symbolic expressions -- ones obtained for set-equivalent queries -- are indeed made equivalent in the structure. It is then we can say that our provenance model captures the ``essence of computation"
defined by the queries, rather than the query structure; the axioms will also allow for optimizations that will be the subject of subsequent sections.

\subsection{Algebraic Structure}
\label{sec:algebStruct}
In a similar vein to the semiring construction of \cite{GKT-pods07},
we start with a basic set of annotations $X$. These could be thought
of as identifiers, which in our case will be associated not only
with tuples but also with individual queries, and propagated to the
tuples they ``touch". We then introduce a structure called $UP[X]$
(``UP" standing for updates) as follows. As a most general
structure, we will start by using six algebraic operations (we later show that five operations are sufficient): $+_I$ and $-_D$
which will serve as abstract operations to capture provenance for
{\em insertion} and {\em deletion} respectively; $-_M$ that will be
used in the context of {\em modification}, to capture the original
tuple (before modification); and $+_M$ and $\cdot_M$ that will be
used for the tuple after modification. Last, we will use $+$ (and
$\Sigma$ for summation over a set), to capture disjunction
originated in the query. 

We also introduce a unique element denoted as $0$, that intuitively
will be used to denote an absent tuple, when used as tuple
annotation, or the fact that an updated query has not taken place,
used as query annotation. Expressions in $UP[X]$ are then comprised
of any combination of elements in $X \cup \{0\}$ using these 
operations; we will sometime refer to such expressions as {\em
formulas}.

We still keep these operations
abstract, in that we do not impose any concrete semantics or further equivalences; we will do both later. 

\paragraph*{Annotated Relations} Let $R$ be a (standard) relation schema and let $tup(R)$ be the set of all tuples conforming to $R$. Given a set of annotations $X$, we use the term $UP[X]$-relation $R$ to denote a function from $tup(R)$ to $UP[X]$. The set of all tuples not mapped to $0$ is called the support of $R$ (we will also say that they are ``in" $R$). This means that $R(t)$ is the provenance annotation (intuitively, at this point, an identifier or meta-data) of a tuple $t$, and if this annotation is non-zero then $t$ is said to be in $R$ (later, when we map annotations to values, it will be useful to map an annotation to $0$, to capture, e.g., tuple deletion). A set of $UP[X]$-relations (associated with a schema, in the standard sense) is an $UP[X]$-database.

\paragraph*{Annotated Update Queries and Transactions} We include an {\em annotation} as part of update query specifications. Intuitively, this annotation may stand for an
identifier of the query, or any other meta-data associated with it. We fix a set $P$ of symbols to
be used as query {\em annotations} and attach them to the heads of queries. For instance, the head of a
provenance-aware insertion query has the form $R^{+,p}(\tup{u})$:-,
where $p \in P$ is the annotation; similarly for deletion and
modification. For example, the query
$Products^{+,p}(\text{``Lego bricks", ``Kids", }\$90)\text{:-}$ is an annotated insertion query with the
annotation being $p$. Similarly, we will use $T^p$ to denote a transaction $T$ annotated by $p$ (i.e., its queries are annotated by $p$).

{\bf Provenance for hyperplane queries.} We are now ready to define provenance for queries. In what follows, let $R$ be an $UP[X]$-relation. We consider different types of update queries $Q$, and use $R'$ for the $UP[X]$-relation that is the result of applying $Q$ to $R$. For example, 
$R$ may be annotated by basic annotations (identifiers) and $R'$ by annotations capturing the computation; but the framework is compositional, so it may be the case that $R$ is already annotated by more complex annotations.  The resulting $UP[X]$-relation is as follows.

%
%
%
%
%
%

\begin{itemize}
    \item If $Q \equiv R^{+,p}(t)$:- then we define $R'(t)= R(t) +_{I} p$, and for each $t' \neq t$ we define $R'(t')=R(t')$.

    \item If $Q \equiv R^{-,p}(\tup{u})$:-, then
    $R'(t)= R(t) +_{D} p$ for each tuple $t \in R,~t\vDash \tup{u}$  and $R'(t')=R(t')$ otherwise.

    \item If $Q \equiv R^{M,p}(\tup{u_1},\tup{u_2})$:-, then $R'(t_1)= R(t_1) -_M p$ for each
    $t_1 \in R~t_1\vDash \tup{u_1}$ and $R'(t_2) = R(t_2)+_{M}( (\sum_{t_1\rightsquigarrow t_2 } R(t_1))\cdot_M p)$, for each $t_2$ s.t. $\exists t_1\vDash \tup{u_1}~t_1 \rightsquigarrow t_2$. The operator $\sum$ here stands for a disjunctive operator associated to the query. In particular, it is different than $+_M$ and $+_I$. Otherwise $R'(t) = R(t)$.
%
%
%
%
\end{itemize}

%
%
%

Note that each algebraic operator in the provenance expression is designed to capture provenance for a query operator. This correspondence will manifest itself in our equivalence axioms below. For now, we only add a special treatment for the $0$ value we have introduced (these will be referred to as ``zero-related axioms" below). Recall that if $t\not\in  R$ then $R(t) = 0$. Thus we define $\forall a\in X$
\begin{itemize}
    \item $0~op~a = 0$ if  $op\in \{-_M, -_D\}$
    \item $0~op~a = a$ if  $op\in \{+_M, +_I\}$
    \item $a~op~0 = a$ for $op\in \{+_I, +_M, -_M, -_D\}$
    \item $a \cdot_{M} 0 = 0 \cdot_M a = 0$
\end{itemize}

Intuitively, if $t\not\in R$, then deleting or modifying $t$ does
not change $R$, and $t$ remains absent from $R$, thus $0~op~a = 0$
if  $op\in \{-_M, -_D\}$.  The existence of an inserted tuple
$t\not\in R$ by a query annotated with $a$ depends only on $a$ and
thus $0+_Ia = a$. Similarly for an updated tuple $t\rightsquigarrow
t'$ for $t'\not\in R$. In a way, the righthand element of the
operators $+_I, +_M, -_M$ and $-_D$ may be interpreted as a
condition for the update, i.e., if the condition is 0, the update
did not take place. Therefore $a~op~0 = a$ for $op\in \{+_I, +_M,
-_M, -_D\}$. Finally, the expression $a\cdot_{M} b$ is used to
capture the fact that a tuple annotated by $a$ is updated by a query
annotated by $b$ to produce an updated tuple. If $a = 0$, the tuple
is not in the database; if $b = 0$ the query has not taken place. In
both cases the updated tuple was not generated, thus $a \cdot_{M} 0
= 0 \cdot_M a = 0$.

We note that, in our setting, for different use-cases it is possible to assign variables the values $1$ or $0$ (as we demonstrate in Section \ref{sec:appExample}). Then, for example, starting from the expression $p_1 +_M (p_2 \cdot_M,p)$, the assignment of the value $1$ to $p$ results in the expression $p_1 +_M p_2$. By further assigning the value $0$ to $p_2$ we obtain the expression $p_1$.


\begin{example}\label{ex:prov}
    Reconsider the database fragment shown in Figure \ref{fig:initTable}, and the annotated update query:
\begin{multline*}
    Products^{M,p}(\text{``Kids mnt bike"}, a, b, \text{``Kids mnt bike", ``Bicycles"}, b)\text{:-}
    \end{multline*}
    By applying the query, the tuple (``Kids mnt bike", Sport, \$120), annotated by $p_1$ and the tuple (``Kids mnt bike", Kids, \$120), annotated by $p_3$
    are updated to (``Kids mnt bike", Bicycles, \$120), which is not in the database, thus annotated by $0$. As a result the new
    tuples
    annotations are $p_1 -_M p$, $p_3 -_M p$ and
    $0 +_M (p_1+ p_3)\cdot_{M}p = (p_1+ p_3)\cdot_{M}p$ respectively.

\end{example}

%

{\bf Provenance of a transaction.} For a given transaction, we annotate it -- i.e., all of its update queries -- with an annotation $p$ (a single annotation is used per transaction, reflecting the grouping of queries to a transaction). We apply the queries in the transaction one by one, using the above definitions to compute the provenance of tuples they ``touch": the $i$'th update is applied on the annotated database obtained from applying the first $i-1$ updates.

\begin{example}\label{ex:transaction}
	Figure \ref{fig:trans1} depicts an example of a transaction over the database fragment given in Figure \ref{fig:initTable}.
    The resulting database from the transaction includes the tuple $Products(\text{``Kids mnt bike", ``Kids"}, \$120)$ with the provenance annotations $p_3 -_M p$ due to the first query. The annotation of the tuple $Products(\text{``Kids mnt bike", ``Sport"}, \$120)$ is $(p_1 +_M (p_3\cdot_{M} p))-_M p$, where the part in the parentheses is the result of the first query. Finally, the annotation of the tuple $Products(\text{Kids mnt bike, Bicycles}, \$120)$ is $0+_M ((p_1 +_M (p_3\cdot_{M} p))\cdot p)$, where the sub-expression $p_1 +_M (p_3\cdot_{M} p)$ comes from the provenance annotation of the tuple $Products(\text{``Kids mnt bike", ``Sport"}, \$120)$ after the execution of the first query.

\end{example}

\begin{figure}
	
	\begin{subfigure}[b]{\linewidth}
		\footnotesize{
		\begin{multline*}
		Products^{M,p}(\text{``Kids mnt bike", ``Kids"}, c, \text{``Kids mnt bike", ``Sport"}, c)\text{:-}
		\end{multline*}
		\vspace{-0.5cm}
		\begin{multline*}
		Products^{M,p}(\text{``Kids mnt bike", ``Sport"}, c, \text{``Kids mnt bike", ``Bicycles"}, c)\text{:-}
		\end{multline*}
		\vspace{-0.5cm}
		\caption{Transaction $T_1$}
		\label{fig:trans1}
}
	\end{subfigure}%

\begin{subfigure}[b]{\linewidth}
	\footnotesize{
	\begin{multline*}
	Products^{M, p}(\text{``Kids mnt bike", ``Kids"}, c, \text{``Kids mnt bike", ``Bicycles"}, c)\text{:-}
\end{multline*}\vspace{-0.5cm}
\begin{multline*}
Products^{M, p}(\text{``Kids mnt bike", ``Sport"}, c, \text{``Kids mnt bike", ``Bicycles"}, c)\text{:-}
\end{multline*}
	\caption{Transaction $T'_1$}
	\label{fig:trans'}
}
\end{subfigure}%
\bigskip

\begin{subfigure}[b]{\linewidth}
\centering
\footnotesize{
	$Products^{M,p'}(a, \text{ ``Sport"}, c,a,\text{ ``Sport"},50)\text{:-}$

	\caption{Transaction $T_2$}
	\label{fig:trans2}
}
\end{subfigure}%
	\vspace{-0.2cm}
	\caption{Transactions}
	\label{fig:trans}
		\vspace{-0.5cm}
\end{figure}
%
%

\subsection{Algebraic Axiomatization}\label{sec:axioms}

The operations we have introduced so far lead to a very abstract
notion of provenance tracking which essentially requires full
tracking of the operation of the update queries that took place, without
allowing for any simplifications.


A fundamental question is what
simplifications can take place, while still capturing the ``essence"
of updates that took place? To this end, we note that \cite{vianu} has introduced a sound and
complete axiomatization of set equivalence for update queries. Combining this axiomatization with our basic
provenance definition, we obtain a set of equivalence axioms over expressions in
$UP[X]$.  We next exemplify the derivation of axioms in our structure based on \cite{vianu}:

\begin{example}
    Based on \cite{vianu}, the following transactions are equivalent

    \begin{center}
        \begin{tabular}{l c l}
        \begin{tabular}{ l c l }
            $R^{M,p}(\tup{u_1},\tup{u_2})$:-\\
            $R^{-,p}(\tup{u_2})$:-
        \end{tabular}
        &
        $\sim$
        &
        \begin{tabular}{ l c l }
            $R^{-,p}(\tup{u_1})$:-\\
            $R^{-,p}(\tup{u_2})$:-
        \end{tabular}
    \end{tabular}
    \end{center}
    
    Note that $\forall t_1\vDash \tup{u_1}$ the provenance expression after the transaction on the left is $R(t_1)-_{M}p$ and after the transaction on the right, $R(t_1)-_{D}p$, thus $a-_{D}b = a-_{M}b$, i.e.,  $-_M$ and $-_D$ are equivalent, and therefore, from now on we use ``$-$" to denote both. Furthermore, $\forall t_2\vDash \tup{u_2}$, from the left transaction we obtain the expression $\Big(R(t_2)+_M\Big(\big(\sum_{\substack{t_1\vDash \tup{u_1}\\ t_1\rightsquigarrow t_2}}R(t_1)\big)\cdot_{M}p\Big)\Big)~-_D~p$, and from the right transaction the expression
    $R(t_2)-_{D}p$. The sum in the first expression represents the set of tuples that are updated into a single tuple. In case there is only one such tuple, it contains a single element, and thus we obtain that $\Big(a +_{M}(b\cdot_M c)\Big)- c = a-c$ for all $a,b$ and $c$.
    

\end{example}

We simplified the axioms and removed redundancies to obtain the  set of equivalence axioms shown in Figure \ref{fig:axioms}.
\begin{figure}
        \input{axioms}
		\caption{Axioms} \label{fig:axioms}
\end{figure}

Note that we have introduced the minimal set of axioms based on \cite{vianu}. When specializing into concrete structures (see below), one may impose further reasonable axioms such as commutativity of $+_{I}$.

A formula can be {\em rewritten} into another formula by applying a sequence of axioms. This rewriting is bidirectional, thus forming an equivalence relation: two formulas $\phi_1$ and $\phi_2$ of our update algebraic structure are equivalent if and only if there is a sequence axioms such that $\phi_1$ can rewritten into $\phi_2$ by using the axioms. We denote it by $\phi_1 \equiv_{UP[X]} \phi_2$.

\subsection{Preserving Provenance Under Set Equivalence}
\label{sec:setEq}

We next state the main property of our construction: two transactions yield the same provenance-aware result if and only if they are set-equivalent. We first define equivalence of transaction under our provenance-aware semantics:

\begin{definition}

    We say that two UP[X]-relations $R,R'$ are $UP[X]$-equivalent, and denote $R \equiv_{UP[X]} R'$, if for every tuple $t$ we have that $R(t) \equiv_{UP[X]}R'(t)$\footnote{Note that in particular, $UP[X]$ equivalence implies that the two relations include the same set of tuples.}. We further say that two UP[X]-databases $D,D'$ are $UP[X]$-equivalent (denote $D \equiv_{UP[X]} D'$) if there is an isomorphism between the relation names in $D$ and $D'$ so that matching relations are $UP[X]$-equivalent.

    Finally, we say that two annotated transactions $T^{p}_1$ and $T^{p}_2$ are $UP[X]$-equivalent, and denote $T^{p}_1 \equiv_{UP[X]} T^{p}_2$ if for every $UP[X]$-database $D$, we have that $T_1^p(D) \equiv_{UP[X]} T_2^p(D)$.
\end{definition}

We further say that for non-annotated transactions $T_1, T_2$, they are set-equivalent, and denote by $T_1 \equiv_{B} T_2$, if for every database $D$, we
have that $T_1(D) \equiv T_2(D)$, where ``$\equiv$'' now stands for standard isomorphism between the databases.

We are now ready to state the following result:

\begin{proposition}\label{prop:equivalece}
    For every two transactions $T_1,T_2$ we have that $T_1 \equiv_{B} T_2$ if and only if $T_1^p \equiv_{UP[X]} T_2^p$.
\end{proposition}
\begin{proof}(sketch)
By definition, $UP[X]$-equivalence implies set-equivalence, thus one direction is trivial. For the other direction, the completeness of axioms from \cite{vianu} guarantees there is a sequence of such axioms   whose application transforms $T_1$ into $T_2$. The proof is then by induction, where for each individual axiom we apply  a corresponding axiom(s) to the provenance. 
\end{proof}

\begin{example}\label{ex:eqTrans}
	Reconsider the database fragment given in Figure \ref{fig:initTable}. According to \cite{vianu}, by modification axiom 2, the  transaction $T_1$ given in Figure \ref{fig:trans1} is set-equivalent to the transaction $T'_1$ in Figure \ref{fig:trans'} .
 	The effect of both is that the tuples $Products^{M}(\text{``Kids mnt bike", ``Kids"}, \$120)$ and \\
 	$Products^{M}(\text{``Kids mnt bike", ``Sport"}, \$120)$ are updated into a single tuple \\
 	$Products^{M}(\text{``Kids mnt bike", ``Bicycles"}, \$120)$. Indeed, the provenance expressions obtained by both are equivalent. The provenance of the tuple \\
 	$Products(\text{``Kids mnt bike", ``Kids"}, \$120)$ is $p_3 -p$, in both cases. The annotation of the tuple $Products(\text{``Kids mnt bike", ``Sport"}, \$120)$ using the latter transaction is $p_1 -p$ and is equivalent to $(p_1 +_M (p_3\cdot_{M} p))-p$ by axiom 2. Last, the annotation of the tuple $Products(\text{Kids mnt bike, Bicycles}, \$120)$ in the database obtained by $T'_1$ is $(0+_M (p_3 \cdot_{M} p)) +_M (p_1\cdot_{M} p)$. By axiom 3 (and using $a = 0$, $I=S_1=\{p_3\}$, $\sum^n_{i=1} b_i = p_1$) it is equivalent to $0+_M ((p_1 +_M (p_3\cdot_{M} p))\cdot p)$, which is 
	the provenance of this tuple obtained by $T_1$ as shown in Example \ref{ex:transaction}, and thus the databases resulting by the two transactions are equivalent.
\end{example}

\paragraph*{Sequence of transactions} We next demonstrate our construction for a sequence of transactions, where each transaction is annotated using a different provenance annotation. 
\begin{example}\label{ex:trSeq}
	Consider a sequence of two transaction $T_1, T_2$, shown in Figures \ref{fig:trans1} and \ref{fig:trans2} resp. Intuitively, the transaction  $T_2$ updates the price of all the products in the ``Sport" category to $\$50$.  Note that the provenance annotation of $T_2$ is $p'$. The database resulting by the application of this sequence on the database shown in Figure \ref{fig:initTable}, contains (among others) the tuples shown in Figure \ref{fig:transOutput}.
\end{example} 

\begin{figure}
	\centering

		\begin{tabular}{llcl}
			\cline{1-3}
			Product     & Category & Price &  \\ \cline{1-3}
			Kids mnt bike    & Sport & \$50 & $0 +_M ((p_1 +_M (p_3\cdot_{M} p))- p)\cdot_M p'$   \\ 
			Tennis Racket & Sport & \$50 & $0+_M(p_2\cdot_M p')$    \\ \cline{1-3}
	
		\end{tabular}

		\caption{Transaction Output (partial)}
		\label{fig:transOutput}

\end{figure}


As expected, two equivalent sequences of transactions yield equivalent provenance expressions associated with their output tuples:

\begin{example}\label{ex:eqSeq}
	Consider the transactions $T_1, T'_1$ and $T_2$ from Figure \ref{fig:trans}. The provenance expression generated for each tuple $t$ in the database by sequence $T_1, T_2$ is equivalent to the provenance of $t$ generated by the sequence $T'_1, T_2$. For instance, the annotation of the tuple $Products(\text{``Kids mnt bike", ``Sport"}, \$50)$ using the latter sequence is $0+_M(p_1 -p)\cdot_M p'$ and is equivalent to $0+_M ((p_1 +_M (p_3\cdot_{M} p))-p)\cdot_M p'$ by axiom 2. Note that we may further simplify both expressions by removing the $0$. 
	
\end{example}

{\bf Comparison with MV-semirings \cite{glavic}.}
There exists a previously proposed algebraic provenance model for update queries,
called MV-semirings \cite{glavic}. 
This model is an extension of the semiring framework, in the sense that for every
semiring $K$, the corresponding MV-semiring $K^\nu$ is introduced. The elements of such a semiring are symbolic expressions over
elements from $K$, version annotations, and semiring operations
where the structure of an expression encodes the derivation history
of a tuple. For instance, $\mathbb{N}[X]^\nu$ is the MV-semiring
corresponding to the provenance polynomials semiring
$\mathbb{N}[X]$. Using this most general $\mathbb{N}[X]^\nu$ MV-semiring, each
tuple is annotated by a provenance expression consisting of
variables which represent identifiers of freshly inserted tuples,
and version annotations that encode the sequence of updates that
were applied to the tuple. The version annotation
$X^{id}_{T,\nu}(k)$ denotes that operation $X$ ($X$ may be one of
$U$, $I$, $D$, or $C$, which stand for update, insert, delete or
commit respectively) was executed at time $\nu-1$ by transaction
$T$, where $k$ is the annotation of the tuple before the update and
$id$ is the identifier of the affected tuple.

Since an MV-semiring counterpart is defined for every semiring, this model is applicable in 
settings beyond those addressed here, notable including support for bag semantics. Further applications such ones pertaining to concurrency are also developed in \cite{glavic}. For such applications, and by design, the model of \cite{glavic} does not satisfy a counterpart of our Proposition \ref{prop:equivalece}: more details on the specific of the transaction that took place are recorded, and so equivalent transactions may yield non-equivalent expressions in the MV-semiring:


 
   \begin{example}\label{ex:compToGlavic}
 	Consider the equivalent transactions sequences  from Examples \ref{ex:eqSeq}.
 	Using the MV-semiring model, applying the two transactions to
 	the database given in Figure \ref{fig:initTable}
 	results in different provenance expressions. For instance, if the provenance annotations satisfy
 	$p_3 = I^1_{T,2}(x_1)$, then the provenance of the tuple $Products(\text{Kids mnt bike, Bicycles}, \$120)$
 	after applying the first transactions sequence contains an expression of the form  $U^3_{T_2,5}(U^2_{T_1,4}(U^1_{T_1,3}(I^1_{T,2}(x_1))))$
 	while the provenance annotation after applying the second transaction contains
 	an expression of the form $U^2_{T_2,4}(U^1_{T'_1,3}(I^1_{T,2}(x_1)))$.
 	
\end{example}

We have highlighted the theoretical appeal of equivalence-invariance that holds for our model but not for \cite{glavic}; in Section~\ref{sec:provmin} we will show that it also allows to optimize provenance representation, and will further show its practical impact in the experiments. In this context, we note that \cite{glavic} further defines an operation called \textsc{Unv} that intuitively 
 removes the embedded history from  the provenance (the parallel of our ``transaction annotations"), while keeping information coming from the underlying semiring $K$ (the parallel of our ``tuple annotations"). The resulting provenance obtained by applying \textsc{Unv} is then equivalence-invariant, but it does not include sufficient information to, e.g., examine the effect of transaction abortion, assign trust values to transaction queries (see Section \ref{sec:app}, in particular Examples \ref{ex:transAbortion}, and the parts of the discussions on access control and certifications pertaining to transaction annotations) or other retroactively reason about meta-data associated with transaction's queries (in contrast to the data).
 
\begin{example}	
\label{exp:unv} 	
 	Applying the \textsc{Unv} operation to either expressions in Example \ref{ex:compToGlavic} yields the same result: $x_1$, reflecting the relevant tuple from the input database (and in general multiple such tuples and their combination) but not the annotations of update queries that took place. 
 \end{example}

%% file: axioms.tex
\setcounter{equation}{0}
\begin{equation}\label{modificationAxiom1_2}
\Big(a+_{M}(b\cdot_{M}c)\Big)+_{M}(d\cdot_{M}c) = \Big(a+_{M}(d\cdot_{M}c)\Big)+_{M}(b\cdot_{M}c)
\end{equation}
\begin{equation}\label{modificationAxiom2_1}
\Big(a +_{M} (b \cdot_{M} c)\Big) - c = a- c
\end{equation}
\begin{equation}
\begin{split}\label{modificationAxiom2_2}
&\mathrm{Let ~}I\mathrm{~be~a~set~of~provenance~expressions~and~}\{S_1,...,S_n\}\\& \mathrm{~be~a~partition~of~}I:\\  
&\big(a+_{M} ((\sum_{c \in I} c) 
\cdot_M d)\big) +_M \big((\sum_{i=1}^n b_i) \cdot_{M}d\big) = \\
&\quad\quad a+_{M}
\Big(\Big(\sum_{i=1}^n(b_i +_M ((\sum_{c \in S_i}c) \cdot_{M} d)\big)\Big) \cdot_{M} d\Big)
\end{split}
\end{equation}

\begin{equation}\label{modify-delete_axiom_7_1}
(a -b)-b = a-b
\end{equation}
\begin{equation}\label{modify-delete_axiom_7_2}
a+_M\Big(\big(\sum_i (b_i- c)\big)\cdot_{M}c\Big) = a
\end{equation}
\begin{equation}\label{modify-insert_axiom_9}
\big(a+_{M}(b\cdot_M c)\big) +_{I} c = (a+_{I}c)+_{M} (b\cdot_M c)
\end{equation}
\begin{equation}\label{modify-insert_axiom_10_1}
(a+_{I}b)-b = a-b
\end{equation}
\begin{equation} \label{modify-insert_axiom_10_2}
a+_{M} \big((b +_{I}c)\cdot_{M}c\big) = (a+_{I} c)+_{M}(b\cdot_{M}c)
\end{equation}
\begin{equation} \label{modify-insert_axiom_11}
\big(a+_{M} (b\cdot_{M}c)\big) +_{I} c = a+_{I} c
\end{equation}
\begin{equation}\label{modify-insert_axiom_13}
(a- b)+_{I} b = a+_{I} b
\end{equation}
\begin{equation}
\label{split_axiom}
a+_{M} (\sum_i b_i + \sum_j  d_j)\cdot_M c) = (a+_{M} (\sum_i b_i \cdot_M c))+_{M} (\sum_j d_j \cdot_M c)
\end{equation}
\begin{equation}\label{interchange_axion_1}
(a- b)+_M (c\cdot_{M}b) = (a- b)+_M \Big(\big((d- b)+_M (c\cdot_{M}b)\big)\cdot_{M}b\Big)
\end{equation}

%% file: provUse.tex
\section{Applications}
\label{sec:app}


We next demonstrate the usefulness of the introduced structure through a concrete semantics assigned to the operators. As we shall illustrate, the general axioms that we have derived above can guide the design of such semantics: care is needed in designing them so that they fit the application of interest, while provenance is still preserved through transactions rewriting.

Each concrete semantics is represented by tuple $(\mathcal{K}, +^{\mathcal{K}}_M, \cdot^{\mathcal{K}}_M, -^{\mathcal{K}} , +^ {\mathcal{K}}_I, +^{\mathcal{K}}, 0^{\mathcal{K}})$ where $\mathcal{K}$ is a set of provenance annotations, and $+^{\mathcal{K}}_M$, $\cdot^{\mathcal{K}}_M$, $-^{\mathcal{K}}$ and $+^ {\mathcal{K}}_I$ are concrete operation over the values in $\mathcal{K}$. We call such tuple Update-Structure.

An important principle underlying the semiring-based
provenance framework is that one can compute an ``abstract" provenance representation and then ``specialize" it in any domain. This ``specialization" is formalized through
the use of semiring homomorphism. To allow for a similar
use of provenance in our setting, we extend the notion of
homomorphism to Update-Structures.


\begin{definition}
Let $S_1 = (\mathcal{K}_1, +^{\mathcal{K}_1}_M, \cdot^{\mathcal{K}_1}_M, -^{\mathcal{K}_1} , +^{\mathcal{K}_1}_I, +^{\mathcal{K}_1},  0^{\mathcal{K}_1})$ and\\ $S_2 = (\mathcal{K}_2, +^{\mathcal{K}_2}_M, \cdot^{\mathcal{K}_2}_M, -^{\mathcal{K}_2} , +^{\mathcal{K}_2}_I, +^{\mathcal{K}_2},  0^{\mathcal{K}_2})$ be two Update-Structures. An homomorphism is a mapping $h: S_1 \mapsto S_2$ such that
\begin{center}
	\begin{tabular}{ll}
	$h(a +^{\mathcal{K}_1}_M  b)= h(a) +^{\mathcal{K}_2}_M h(b)$ &$	h(a \cdot^{\mathcal{K}_1}_M  b)= h(a) \cdot^{\mathcal{K}_2}_M  h(b) $\\
	$h(a -^{\mathcal{K}_1}  b)= h(a) -^{\mathcal{K}_2} h(b)$&
	$h(a +^{\mathcal{K}_1}_I  b)= h(a)+^ {\mathcal{K}_2}_I h(b)$ \\
	$h(a +^{\mathcal{K}_1}b) = h(a)+^{\mathcal{K}_2}h(b)$ &
	$h(0^{\mathcal{K}_1})  = 0^{\mathcal{K}_2}$
\end{tabular}
\end{center}
\end{definition}

Crucially, we may show that provenance propagation commutes with homomorphisms. We use $T(D)$ to denote the database obtained from applying the transaction $T$ on the database $D$, and say that a tuple $t$ in $T(D)$ if $t$ in the resulting database.

\begin{proposition}\label{prop:commutes_with_homo}
    Let $S_1$ and $S_2$ be two Update Structures such that there exists an homomorphism  from $S_1$ to $S_2$.
    Let $D$ be a database instance, $T$ a transaction and $t$ a tuple in $T(D)$. Let $\phi_1(t)$ (respectively $\phi_2(t)$) be the provenance expression of $t$  by $T$ over $S_1$ (respectively $S_2$). We have that $h(\phi_1(t))=\phi_2(t)$. 
\end{proposition}
This property allows us to support applications as exemplified next.

\subsection{Example Semantics}
\label{sec:appExample}

We next highlight multiple semantics of interest and their corresponding algebraic structures.

\paragraph*{Deletion Propagation} 
Consider an
analyst who wishes to examine the effect of \emph{deleting a tuple} from
the input database on the result of a sequence of transactions. This may be done without provenance, by
actually deleting the tuple and re-running the sequence. Alternatively, and much more efficiently, if we have provenance we may 
assign truth values to annotations occurring in it. In particular, deleting a tuple corresponds to assigning False to the tuple annotation.
The provenance semantics that allows for deletion propagation is the following 
\begin{center}
	 \begin{tabular}{c}
	$a +_M  b= a +_I  b = a + b := a \vee b $\\
	 \begin{tabular}{lr}
	 	$a \cdot_M  b:= a \wedge b$ &$a -  b:= a \wedge \neg b$
\end{tabular}
\end{tabular}
\end{center}
Where $0$ corresponds to the Boolean value False.
	
\begin{example}\label{ex:delProp}
	Reconsider the transactions sequence $T_1, T_2$ from Example \ref{ex:trSeq}, and the tuple $t = products(\text{``Tennis Racket"},\text{``Sport"}, \$50)$ annotated by $0+_M(p_2\cdot_M p')$ in the output. The scenario where the tuple  $products(\text{``Tennis Racket"},\text{``Sport"}, \$70)$ is omitted from the initial database corresponds to the valuation that assigns False to $p_2$. With the above semantics, in this case, the tuple $t$ will not appear in the output.
	
\end{example}

\paragraph*{Transaction Abortion} 
The same provenance structure allows to examine the effect of \emph{aborting a transaction}, on the result of a sequence of transactions. Again, a naive way to do it is to re-run the sequence while ignoring the aborted transaction, but 
the same results may be achieved efficiently using the provenance information (as we show in Section \ref{sec:exp}): aborting a transaction corresponds to assigning False to the aborted {\em transaction annotation}.

\begin{example}\label{ex:transAbortion}
	Consider again the transactions sequence from Example \ref{ex:trSeq}. The scenario where the first transaction is aborted corresponds to assigning the truth value False to the variable $p$ in the provenance expression. With the above semantics, the provenance expression of the tuple $Products(\text{``Kids mnt bike", ``Sport"}, \$50)$ is evaluated to True, i.e., if we abort the first transaction we would indeed obtain this tuple in the resulting database. 
\end{example} 


\paragraph*{Access Control} Consider an application that supports different products and prices for different countries (e.g., based on different shipping costs and taxes). Each tuple is annotated with a set of country names, such that a user from country $c$ can see a tuple $t$ only if $t$'s annotation contains $c$. Similarly, {\em transactions are also annotated by sets of countries, so that the transaction annotation defines the set of countries that are affected by the update}. For instance, if a deletion query $q$ deletes the tuple $t$ and $q$'s annotation contains the country $c$, then after the deletion the tuple $t$ is no longer available for users from the country $c$.

This semantics may formally be captured in our framework by defining the following provenance operations:
\begin{center}
	\begin{tabular}{c}
 $a +_M  b = a +_I  b= a + b:= a \cup b$\\
 \begin{tabular}{lr}
 	$a \cdot_M  b:= a \cap b$ & $a -  b:= a \setminus b$
 \end{tabular}
\end{tabular}
\end{center}
defined over the domain of sets (whose individual items are, e.g., country names).

\paragraph{Tuples/Transactions Certification}

Consider an application where tuples/transaction are associated with values from $[0,1]$, reflecting their level of trust. Then given a minimal trust level $L$, we wish to know the result of an execution that involve  only transaction and tuple with trust score that exceeds $L$. This can be done by using annotation of the form $a = (v, r)$, where $a.v\in[0,1]$ in the trust score of the tuple/transaction, and $a.r$ is  ``trusted with respect to $L$" and can be one of $T$ (True), $F$ (False) or $U$ (unknown). For brevity of notation, we then use $trusted(x)$ as a macro for 
($x.r=T$) or ($x.r=U$ and $x.v>L$). The operations are then defined through a Boolean structure over the $trusted$ values (note that their corresponding truth values will not be materialized until assigned concrete trust values to input tuples): 
\begin{align*}
    &a +_M  b = a +_I  b= a + b:=
 	 \begin{cases}
 		(1, T)    & \text{if } trusted(a)\text{ or } trusted(b)  \\
 		(0,F)    & otherwise
 	\end{cases}	 \\
	        &a -  b:= \begin{cases}
	            (1, T)    & \text{if } trusted(a) \text{ and } NOT(trusted(b)) \\
	            (0, F)    & otherwise
	        \end{cases}	      \\
        &a \cdot_M  b:= \begin{cases}
        	(1, T)   & \text{if } trusted(a)\text{ and } trusted(b) \\
        	(0,F)   & otherwise
        \end{cases}	      
\end{align*}

We may show that all proposed structures satisfy the axioms from Section  \ref{sec:axioms} (proof omitted for lack of space).

\subsection{From semirings to $UP[X]$-operators}
As discussed above, it is commonplace to define algebraic provenance through semirings. We next show how to transform a commutative semiring -- given that it satisfies some natural constraints -- into an $UP[X]$ structure that can be used for provenance in the presence of update queries.
\begin{theorem}\label{th:semring}
    Let $(K, +_K, \cdot_{K}, 0, 1)$ be a commutative semiring that satisfies $a+_K 1 = 1$ and $a \cdot_K a = a$, then the set of elements $X = K$, with the operators $+_M, +_I, \cdot_{M}$
    defined as follows: $\forall a, b \in X$:
    \begin{center}
        \begin{tabular}{ccc}
          $a +_M b= a +_K b$&
        $a +_I b = a +_K b$&
        $a\cdot_M b  = a\cdot_K b$
    \end{tabular}
    \end{center}

and any $-$ operator that satisfies the axioms 2, 4, 5, 7, 10 and 12 from Section \ref{sec:axioms} with respect to the semiring $+$ and $\cdot$ operators, is an $UP[X]$ structure.

\end{theorem}

The proof is by carefully going through all axioms and is omitted for lack of space.   

\begin{example}
\label{ex:minus}
	Recall the access control example from Section \ref{sec:appExample}. The corresponding semiring is $(\mathcal{P(C)}, \cup, \cap, \emptyset, \mathcal{C})$ where $\mathcal{C}$ is the set of all countries and $\mathcal{P(C)}$ is the power set of $\mathcal{C}$. Note that this is a commutative semiring that satisfies $\forall a \in\mathcal{P(C)}~ a \cup \mathcal{C} = \mathcal{C}$ and $a \cap a = a$. Furthermore, by defining the $-$ operator as set-difference we obtain a structure that satisfies the axioms.
	
	The PosBool semiring $(\mathbb{N[B]}, \vee, \wedge, \bot, \top)$ with the minus operator $a-b = a \wedge(\neg b)$ satisfies the axioms as well. The latter is the structure we demonstrate in the deletion propagation example in Section \ref{sec:appExample}.
	
	Interestingly, the monus operator used in \cite{GeertsP10} to capture relational difference does not generally ``work" as minus in our setting. For instance, our Axiom 10 $((a-b)+b = a+b)$ does
	not hold in general for monus.

	
\end{example}

Note that in particular for this construction $+_I$ and $+_M$ are commutative.

%% file: provMin.tex
\section{Efficient Provenance Computation}
\label{sec:provmin}

We next consider the issue of complexity: how large may the provenance be? Can it be efficiently computed alongside query evaluation?

\subsection{Naive Construction}
\label{sec:naive}

A first attempt is to generate provenance by directly using the
definitions. That is, starting from the initial instance, we apply
sequentially the update queries. We compute the provenance of each
tuple after each update using the definitions of Section
\ref{sec:model}. Unfortunately, this
approach incurs an exponential blowup in the transaction length.



\begin{proposition}
    There exists a transaction $T$ and a database $D$ with only two tuples  $t_1$ and $t_2$ such that the provenance of $t_1$ and the provenance of $t_2$ after applying $T$ to $D$ is at least exponential in the number of queries.
\end{proposition}

\begin{proof}(Sketch)
Let $D$ be a relational database with a single unary relation $R$. Let $t_1 = R(a)$ and $t_2 = R(b)$ be the two tuples belonging to $D$.
The transaction is a sequence of two alternating modification queries. The first modifies $t_1$ to $t_2$, denoted $U_{12}$, and the second modifies $t_2$ to $t_1$, denoted by $U_{21}$.
The transaction starts with an update $U_{12}$.
We denote by $P^i(t_j)$, the provenance of $t_j$ after applying    $i$ updates of $T$. By a simple induction, we can prove that
\begin{itemize}
\item $|P^{2\cdot i}(t_2)| =|P^{2\cdot i-1}(t_1)|+3+ |P^{2\cdot i-1}(t_2)|$
\item $|P^{2\cdot i}(t_1)| =|P^{2\cdot i-1}(t_1)|+2 $
\item $|P^{2\cdot i+1}(t_1)| =|P^{2\cdot i}(t_1)|+3+ |P^{2\cdot i}(t_2)|$
\item $|P^{2\cdot i+1}(t_2)| =|P^{2\cdot i}(t_2)|+2 $
\end{itemize}
Therefore, $|P^{2\cdot i}(t_2)|$  is equal to $2\cdot  |P^{2\cdot (i-1)}(t_2)|+8+ |P^{2\cdot (i-1)}(t_1)|$. Thus,  $|P^{2\cdot i}(t_2)|$ is greater than $2^i$.
\end{proof}

Fortunately, we introduce a normal form for our provenance expression which is linear in the database size and the transaction length. Moreover, we prove that this normal form is computable in polynomial time in the  size of the database and the transaction.

\subsection{Normal Form}
\label{sec:normalForm}
For presentation purposes, we represent our provenance expressions
as trees in a classical manner. Figure \ref{fig:provCircuit} depicts
the basic tree representation for each one of the provenance
operations. Any provenance expression obtained by the construction
for the class of ``domain-based" transactions, when applied to an
$X$-database (i.e. a database whose tuple annotations are just
identifiers), can be represented as a composition of the basic
trees.
\begin{figure}
	\centering
	\begin{subfigure}[b]{0.3\linewidth}
		\centering
\begin{tikzpicture}

\Tree [.\node[draw, align=center, inner sep=1pt, text centered]{$+_I$}; $a$ $p$ ]

\end{tikzpicture}
		\caption{$a +_I p$}
		\label{fig:insertCircuit}
	\end{subfigure}
	\begin{subfigure}[b]{0.3\linewidth}
		\centering
		\begin{tikzpicture}

 \Tree [.\node[draw, align=center, inner sep=1pt, text centered]{$-$}; $a$ $p$ ]

\end{tikzpicture}
		\caption{$a - p$}
		\label{fig:deleteCircuit}
	\end{subfigure}
	\begin{subfigure}[b]{0.3\linewidth}
		\centering

		\begin{tikzpicture}

\Tree [.\node[draw, align=center, inner sep=1pt, text centered]{$+_M$}; $a$ [. \node[circle, draw, align=center, inner sep=1pt, text centered]{$\cdot$}; $b$ $p$ ] ]

\end{tikzpicture}
		\caption{$a +_M (b \cdot_{M}p)$}
		\label{fig:modifyCircuit}
	\end{subfigure}
\caption{Tree representation of provenance}
	\label{fig:provCircuit}
\end{figure}
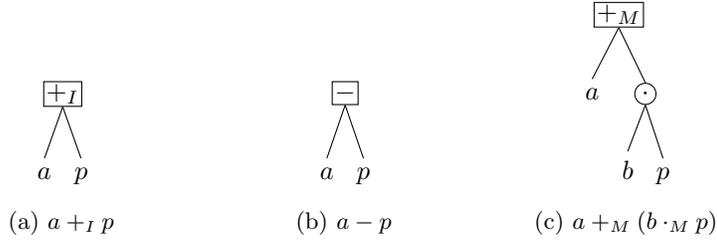

We demonstrate that we can find a normal form of the provenance expression as stated in the following theorem.


\begin{theorem}
\label{lem:minform}
Given a transaction $T^p$, an $X$-database $D$, and $t\in T^p(D)$ with the provenance expression $\phi$. Then, there exists an equivalent provenance expression $\phi' \sim \phi$ such that the tree representation of $\phi'$ has one of the following forms:

    \begin{center}

    		\begin{tikzpicture}

\node at (0,1) {(1)};
\node{$a$};

\begin{scope}[xshift=1cm]
\node at (0, 1) {(2)};
\Tree [.\node[draw, align=center, inner sep=1pt, text centered]{$+_I$}; $a$ $p$ ]
\end{scope}

\begin{scope}[xshift=3cm]
\node at (0, 1) {(3)};
\Tree [.\node[draw, align=center, inner sep=1pt, text centered]{$-$}; $a$ $p$ ]

\end{scope}
\begin{scope}[xshift=5.5cm]
\node at (0, 1) {(4)};
\Tree [.\node[draw, align=center, inner sep=1pt, text centered]{$+_M$}; $a$ [. \node[circle, draw, align=center, inner sep=1pt, text centered]{$\cdot$}; [.\node[draw, align=center, inner sep=1pt, text centered]{$+$}; $b_0$ \quad $b_n$ ] $p$ ] ]
\node at (0.3, -3.1) {$\cdots$};
\end{scope}
\begin{scope}[xshift=10cm]
\node at (0, 1) {(5)};
\Tree [.\node[draw, align=center, inner sep=1pt, text centered]{$+_M$}; [.\node[draw, align=center, inner sep=1pt, text centered]{$-$}; $a$ $p$ ] [. \node[circle, draw, align=center, inner sep=1pt, text centered]{$\cdot$}; [.\node[draw, align=center, inner sep=1pt, text centered]{$+$}; $b_0$ \quad $b_n$ ] $p$ ] ]
\node at (0.4, -3.1) {$\cdots$};
\end{scope}
\end{tikzpicture}

%
%
%
%

    \end{center}

%
Computing $\phi'(t)$ may be performed in polynomial time in the size of $D$ and $T$. Moreover, $\phi'(t)$ can be computed incrementally for each update of the transaction.

\end{theorem}

\begin{proof} (sketch)
The key idea behind the proof is to derive from our axioms a set of operational {\em rules} that manipulate the provenance, shown in Figure \ref{fig:rules}. We may show that the rules are implied by the axioms (but not vice versa), and they guide the generation of a ``normal form".
 Intuitively, in Rule 1 and 2, $a$ is the annotation
associated to the tuple on which the update is applied. Applying an insertion or a deletion overrides the previous updates.
Rules 3 and 8 intuitively state that an update based on an deleted
tuple has no effect and Rule 4 states that an update based on an
inserted tuple is equivalent to inserting the current tuple.  Rules
5, 6 and 7 intuitively allow to ``factorize" successive updates into
a single update.


Then, each update may be handled by applying corresponding rules to the provenance it yields. For instance, for insertion we apply Rule 1 and replace the  provenance  by one of size
    $3$. For deletion, we again obtain size-$3$ expression, this time by applying Rule 2. Modification involves applying the other rules, in a more complex way (details omitted for lack of space). We may show  that after each step, we compute only a linear size formula and that the number of operations performed on this formula
is polynomial in the database and the size of the (prefix of the)
transaction.
\end{proof}




	\begin{figure}
\footnotesize{
	\begin{tabular}{|c |c |}
		\hline
	\begin{tabular}{ c}
Rule 1\\ \begin{tikzpicture}[scale=0.6]
            \Tree [.\node[draw, align=center, inner sep=1pt, text centered]{$+_I$}; \node[draw,dashed,shape border uses incircle,
            isosceles triangle,shape border rotate=90,yshift=-0.1cm, inner sep=1pt] (x){$\tau$}; $p$  ]
            \node (y) at (-0.75,-1.6) {$a$};
            \node (x) at (-0.5,-1.25) {};
            \draw[semithick, -] (x) to  (y);
            \node at (1.5,-0.5) {$\iff$};
            \begin{scope}[xshift=3cm]
            \Tree [.\node[draw, align=center, inner sep=1pt, text centered]{$+_I$}; $a$ $p$  ]
            \end{scope}
            \end{tikzpicture}
        \end{tabular}&
        \begin{tabular}{ c}
            Rule 2 \\ 
            \begin{tikzpicture}[scale=0.6]
            \Tree [.\node[draw, align=center, inner sep=1pt, text centered]{$-$}; \node[draw,dashed,shape border uses incircle,
            isosceles triangle,shape border rotate=90,yshift=-0.1cm, inner sep=1pt] (x){$\tau$}; $p$  ]
            \node (y) at (-0.75,-1.6) {$a$};
            \node (x) at (-0.5,-1.25) {};
            \draw[semithick, -] (x) to  (y);
            \node at (1.5,-0.5) {$\iff$};
            \begin{scope}[xshift=3cm]
            \Tree [.\node[draw, align=center, inner sep=1pt, text centered]{$-$}; $a$ $p$  ]
            \end{scope}

            \end{tikzpicture}
        \end{tabular}
        \\ \hline
        \begin{tabular}{c}
            Rule 3\\ \begin{tikzpicture}[scale=0.6]
            \Tree [.\node[draw, align=center, inner sep=1pt, text centered]{$+_M$}; \node[draw,dashed,shape border uses incircle,
            isosceles triangle,shape border rotate=90,yshift=-0.1cm, inner sep=1pt] (x){$\tau$}; [. \node[circle, draw, align=center, inner sep=1pt, text centered]{$\cdot$}; [.\node[draw, align=center, inner sep=1pt, text centered]{$+$}; [.\node[draw, align=center, inner sep=1pt, text centered]{$-$}; $b_0$ $p$ ] [.\node[draw, align=center, inner sep=1pt, text centered]{$-$}; $b_n$ $p$ ] ] $p$ ]  ]
            \node at (0.4,-3.1) {$\cdots$};
            \node at (2.5,-2) {$\iff$};
            \begin{scope}[xshift=3.5cm, yshift=-1.6cm]
            \node[draw,dashed,shape border uses incircle,
            isosceles triangle,shape border rotate=90,yshift=-0.5cm, inner sep=1pt] (x){$\tau$};
            \end{scope}

            \end{tikzpicture}
        \end{tabular}
        &
        \begin{tabular}{ c}
            Rule 4 \\
            \begin{tikzpicture}[scale=0.6]
            \Tree [.\node[draw, align=center, inner sep=1pt, text centered]{$+_M$}; \node[draw,dashed,shape border uses incircle,
            isosceles triangle,shape border rotate=90,yshift=-0.1cm, inner sep=1pt] (x){$\tau$}; [. \node[circle, draw, align=center, inner sep=1pt, text centered]{$\cdot$}; [.\node[draw, align=center, inner sep=1pt, text centered]{$+$}; [.\node[draw, align=center, inner sep=1pt, text centered]{$+_I$}; $b_0$ $p$ ] \node[draw,dashed,shape border uses incircle,isosceles triangle,shape border rotate=90,yshift=-0.1cm, inner sep=1pt] {$\tau_1$} ;   ] $p$ ]  ]
            \node at (0.4,-3.1) {$\cdots$};
            \node at (2.5,-2) {$\iff$};
            \begin{scope}[xshift=4cm, yshift=-1.6cm]
            \Tree [.\node[draw, align=center, inner sep=1pt, text centered]{$+_I$}; \node[draw,dashed,shape border uses incircle,
            isosceles triangle,shape border rotate=90,yshift=-0.1cm, inner sep=1pt] (x){$\tau$}; $p$  ]
            \end{scope}
            \end{tikzpicture}
        \end{tabular}
    \\ \hline
        \begin{tabular}{ c}
            Rule 5 \\
            \begin{tikzpicture}[scale=0.6]
            \Tree [.\node[draw, align=center, inner sep=1pt, text centered]{$+_M$};
            [.\node[draw, align=center, inner sep=1pt, text centered]{$+_I$}; \node[draw,dashed,shape border uses incircle,
            isosceles triangle,shape border rotate=90,yshift=-0.1cm, inner sep=1pt] (x){$\tau_1$}; $p$ ]
            [.\node[circle, draw, align=center, inner sep=1pt, text centered]{$\cdot$};
            \node[draw,dashed,shape border uses incircle,
            isosceles triangle,shape border rotate=90,yshift=-0.1cm, inner sep=1pt] (x){$\tau_2$};  $p$ ] ] \node at (1.5,-1.2) {$\iff$};
            \begin{scope}[xshift=2.8cm, yshift=-0.5cm]
            \Tree [.\node[draw, align=center, inner sep=1pt, text centered]{$+_I$}; \node[draw,dashed,shape border uses incircle,
            isosceles triangle,shape border rotate=90,yshift=-0.1cm, inner sep=1pt] (x){$\tau_1$}; $p$  ]
            \end{scope}

            \end{tikzpicture}
        \end{tabular}
        &
        \begin{tabular}{ c}

            Rule 6 \\
            \begin{tikzpicture}[scale=0.6]
                \Tree [.\node[draw, align=center, inner sep=1pt, text centered]{$+_M$};
            [.\node[draw, align=center, inner sep=1pt, text centered]{$+_M$}; \node[draw,dashed,shape border uses incircle,
            isosceles triangle,shape border rotate=90,yshift=-0.1cm, inner sep=1pt] (x){$\tau_1$}; [.\node[circle, draw, align=center, inner sep=1pt, text centered]{$\cdot$};
            \node[draw,dashed,shape border uses incircle,
            isosceles triangle,shape border rotate=90,yshift=-0.1cm, inner sep=1pt] (x){$\tau_2$};  $p$ ] ]
            [.\node[circle, draw, align=center, inner sep=1pt, text centered]{$\cdot$};
            [.\node[draw, align=center, inner sep=1pt, text centered]{$+$}; \node[draw,dashed,shape border uses incircle,
            isosceles triangle,shape border rotate=90,yshift=-0.1cm, inner sep=1pt] (y){$\tau_3$};  ]  $p$ ] ]

            \node at (2,-1.2) {$\iff$};

            \begin{scope}[xshift=4cm]
            \Tree [.\node[draw, align=center, inner sep=1pt, text centered]{$+_M$}; \node[draw,dashed,shape border uses incircle,
            isosceles triangle,shape border rotate=90,yshift=-0.1cm, inner sep=1pt] (x){$\tau_1$}; [.\node[circle, draw, align=center, inner sep=1pt, text centered]{$\cdot$};
            [.\node[draw, align=center, inner sep=1pt, text centered]{$+$};  \node[draw,dashed,shape border uses incircle,
            isosceles triangle,shape border rotate=90,yshift=-0.1cm, inner sep=1pt] (x){$\tau_2$};  \node[draw,dashed,shape border uses incircle,
            isosceles triangle,shape border rotate=90,yshift=-0.1cm, inner sep=1pt]
             (x){$\tau_3$};] $p$ ]  ]
            \end{scope}

            \end{tikzpicture}
        \end{tabular}
        
        \\ \hline
        
        \begin{tabular}{ c}
            \\
            Rule 7 \\
            \begin{tikzpicture}[scale=0.7]
            \Tree [.\node[draw, align=center, inner sep=1pt, text centered]{$+_M$}; \node[draw,dashed,shape border uses incircle,
            isosceles triangle,shape border rotate=90,yshift=-0.1cm, inner sep=1pt] (x){$\tau_1$}; [.\node[circle, draw, align=center, inner sep=1pt, text centered]{$\cdot$};
            [.\node[draw, align=center, inner sep=1pt, text centered]{$+$};
            [.\node[draw, align=center, inner sep=1pt, text centered]{$+_M$}; \node[draw,dashed,shape border uses incircle,
            isosceles triangle,shape border rotate=90,yshift=-0.1cm, inner sep=1pt] (x){$\tau_2$}; [.\node[circle, draw, align=center, inner sep=1pt, text centered]{$\cdot$}; \node[draw,dashed,shape border uses incircle,
            isosceles triangle,shape border rotate=90,yshift=-0.1cm, inner sep=1pt] (x){$\tau_3$}; $p$ ]] \node[draw,dashed,shape border uses incircle,
            isosceles triangle,shape border rotate=90,yshift=-0.1cm, inner sep=1pt] (x){$\tau_4$}; ] $p$ ]]

            \node at (2.8,-2) {$\iff$};

            \begin{scope}[xshift=5cm, yshift=-0.5cm]
            \Tree [.\node[draw, align=center, inner sep=1pt, text centered]{$+_M$}; \node[draw,dashed,shape border uses incircle,
            isosceles triangle,shape border rotate=90,yshift=-0.1cm, inner sep=1pt] (x){$\tau_1$}; [.\node[circle, draw, align=center, inner sep=1pt, text centered]{$\cdot$};
            [.\node[draw, align=center, inner sep=1pt, text centered]{$+$};  \node[draw,dashed,shape border uses incircle,
            isosceles triangle,shape border rotate=90,yshift=-0.1cm, inner sep=1pt] (x){$\tau_2$};  \node[draw,dashed,shape border uses incircle,
            isosceles triangle,shape border rotate=90,yshift=-0.1cm, inner sep=1pt] (x){$\tau_3$}; \node[draw,dashed,shape border uses incircle,
            isosceles triangle,shape border rotate=90,yshift=-0.1cm, inner sep=1pt] (x){$\tau_4$}; ] $p$ ]  ]
            \end{scope}
            \end{tikzpicture}

        \end{tabular}
        &
        \begin{tabular}{ c}
            Rule 8 \\
            \begin{tikzpicture}[scale=0.7]
            \Tree [.\node[draw, align=center, inner sep=1pt, text centered]{$+_M$}; \node[draw,dashed,shape border uses incircle,
            isosceles triangle,shape border rotate=90,yshift=-0.1cm, inner sep=1pt] (x){$\tau$}; [. \node[circle, draw, align=center, inner sep=1pt, text centered]{$\cdot$}; [.\node[draw, align=center, inner sep=1pt, text centered]{$+$}; \node[draw,dashed,shape border uses incircle,
            isosceles triangle,shape border rotate=90,yshift=-0.1cm, inner sep=1pt] (y){$\tau_1$}; [.\node[draw, align=center, inner sep=1pt, text centered]{$-$}; $b_0$ $p$ ] ]  $p$ ]  ]
            \node at (0.4,-3.1) {$\cdots$};
            \node at (2.3,-2) {$\iff$};

            \begin{scope}[xshift=3.8cm, yshift=-0.5cm]
            \Tree [.\node[draw, align=center, inner sep=1pt, text centered]{$+_M$}; \node[draw,dashed,shape border uses incircle,
            isosceles triangle,shape border rotate=90,yshift=-0.1cm, inner sep=1pt] (x){$\tau$}; [. \node[circle, draw, align=center, inner sep=1pt, text centered]{$\cdot$}; [.\node[draw, align=center, inner sep=1pt, text centered]{$+$}; \node[draw,dashed,shape border uses incircle,
            isosceles triangle,shape border rotate=90,yshift=-0.1cm, inner sep=1pt] (y){$\tau_1$};  ]  $p$ ]  ]
            \end{scope}

            \end{tikzpicture}
        \end{tabular}\\
    \hline
    \end{tabular}
}
 \caption{Rules for computing the normal form}

    \label{fig:rules}

	\end{figure}
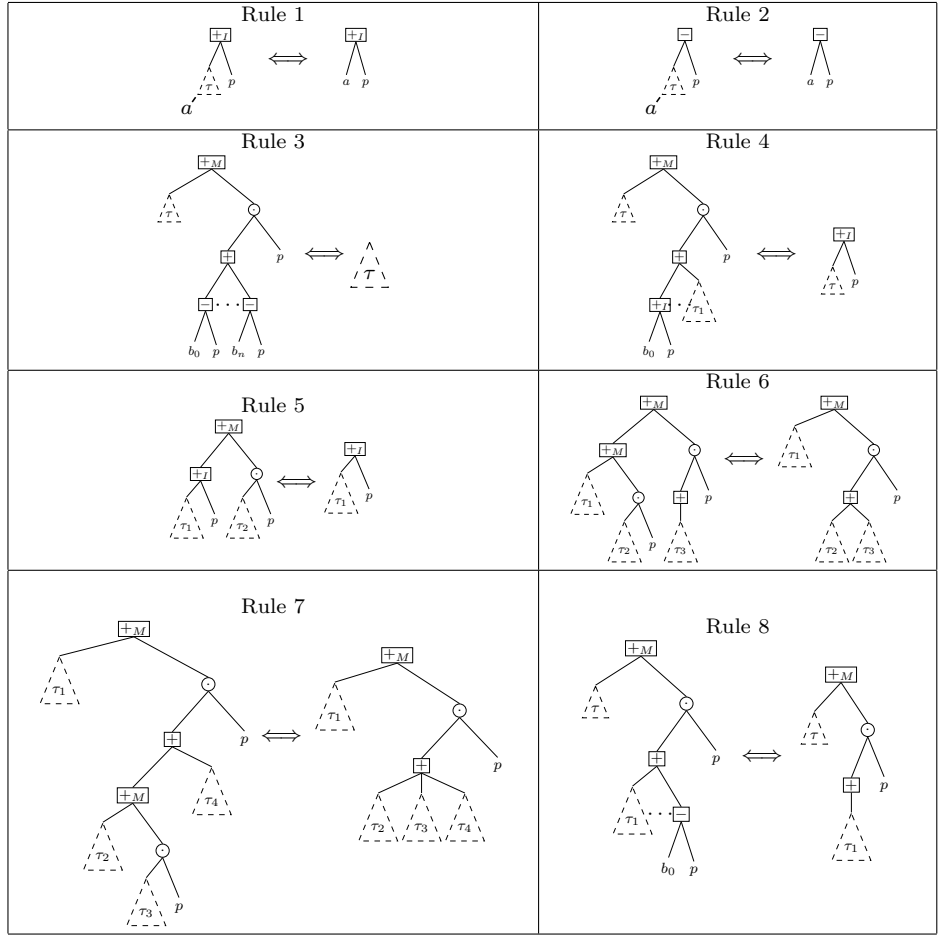

This normal form is still not guaranteed to be minimal, since there is a subtlety pertaining to the possible existence of $0$ in the formula. This may be remedied in post-processing: 
\begin{proposition}
    Let $T^p$ be an annotated transaction, applied to an $X$-database $D$. Let $\phi(t)$ be the normal form provenance expression of a tuple $t$ after applying the transaction $T^{p}$ to $D$.
    Let $\phi'(t)$ be the expression obtained by using the axioms related to $0$ to $\phi(t)$ to minimize it. Then $\phi'(t)$ is unique and a minimized formula.
\end{proposition}

\begin{proof} (sketch)
    We observe that by applying the ``$0$ axioms" (from Section \ref{sec:axioms}) to a normal form formula, we may obtain either (1) a normal form expression,  or (2) $0$ or (3) a formula of the form $\Sigma_i (b_i) \cdot_M p$.
    We can show that none of these expressions is equivalent to any other, and there is
no further concise way of representing neither of them.
\end{proof}

\begin{example}
Consider again the transaction $T_1$ from Figure \ref{fig:trans1} (let $U^p_{1}, U^p_{2}$ denote its first and second query respectively), and the database depicted in Figure \ref{fig:initTable}.
This transaction deals with three tuples:
$t_1 =Products^{M}(\text{``Kids mnt bike", ``Sport"}, \$120)$ with the annotation $p_1$,
$t_2 =Products^{M}(\text{``Kids mnt bike", ``Kids"}, \$120)$ annotated by $p_3$, and 
$t_3 =Products^{M}(\text{``Kids mnt bike", ``Bicycles"}, \$120)$ annotated by $0$.
Normal form is maintained incrementally, in the sense that after each update operation, we examine the provenance expressions of all tuples and, if a particular expression is not in normal form, transform it into one using the rules.    
In our example, after the first update, the provenance of all tuples is already in normal form : $U_1^p(D)(t_3)= p_3 - p$ and $U_1^p(D)(t_1)= p_1 +_M (p_3 \cdot_M p)$.
After the second update, the provenance expressions a of $t_1$ and $t_3$ are no longer in normal form. $T_1^p(D)(t_1)= (p_1 +_M (p_3 \cdot_M p ))- p$ is simplified by using Rule 2, to $p_1 - p$. By using Rule 7,  $T^p(D)(t_3)= 
0 +_M ((p_1 +_M (p_3 \cdot_M p)) \cdot_M p)$
 may be simplified to 
$0 +_M ((p_1+p_3) \cdot_M p)$. 
Further updates, if exist, would apply to these normal forms; if needed their resulting provenance is again transformed to normal forms etc. In this case we have concluded the updates; a post-processing step using the $0$ axioms is applied to the provenance of $T^p(D)(t_3)$ to obtain $(p_1+p_3) \cdot_M p$.
\end{example}

%% file: exp.tex
\section{Experimental Evaluation}
\label{sec:exp}


We have conducted experiments whose main goals were examining (1)
the scalability of the approach with respect to the number of updates in terms of time and memory
overhead, (2) the usefulness of the resulting provenance, assessed by measuring the time it takes to assign values to provenance annotations occurring in the expression, (3) the effectiveness
of our provenance normal form representation which in turn is based
on our provenance equivalence axiomatization, and (4) comparison with the previously proposed model of \cite{glavic}.


We used Python 3 to implement our provenance framework for an in-memory database.
This is a simple proof-of-concept, with no indices, thus each update requires a full scan of the database.
We use a hashmap between tuples and their annotations,
allowing random access to the annotation given the tuple.
The experiments were executed on Windows 10, 64-bit, with
8GB of RAM and Intel Core i7-4600U 2.10 GHz processor. Each experiment was executed 5 times and we report the average result. 

%
%
%
%
%
%
%
%

\subsection{Setup: Benchmarks and baselines}
\label{sec:benchmarks}

We have examined our solutions using two benchmarks: TPC-C \cite{tpcc} is an on-line transaction processing benchmark, including 
update-intensive
transactions, that simulate the activity of complex on-line
transaction processing application environments. Its underlying database
consists of nine tables and is populated with initial data of about
$2.1$M tuples. For our experiments, we used the Python
open source implementation of the benchmark from \cite{tpccimpl} to
generate transactions logs with up to $1966$ update queries, and executed the log using our in-memory database implementation with provenance support. Additionally, we have generated a simple synthetic dataset
populated with $1$M tuples, with randomly generated
values from a fixed domain using a uniform distribution. We
generated sequences of update queries of varying length. The type of
query (insertion, deletion or update) was randomly selected with
uniform distribution; the query parameters (e.g., which
tuples are modified and how) were selected at random from a fixed
domain; deletion and modification queries perform a selection
over a numeric column. 
\vspace{-5mm}
\paragraph*{Compared Algorithms and Baselines} In all experiments we have measured the performance of both of our constructions: (1) the naive approach of Section \ref{sec:naive} that simply generates provenance according to its definition in Section \ref{sec:algebStruct}, and makes no use of neither the normal form nor axioms (labeled ``No axioms" in the graphs); and (2) the more efficient provenance generation method of Section \ref{sec:normalForm} based on the normal form (labeled ``Normal form"). Two baselines that we have compared to are (1) ``No provenance", i.e., vanilla evaluation of the transactions without provenance support, and (2) in dedicated experiments, the provenance model of MV-semirings \cite{glavic} discussed above.  

As explained above, we have also measured the time it takes to {\em use} provenance, for the applications in Section \ref{sec:app}. As is the case with semiring provenance \cite{GKT-pods07}, using provenance for any of these (or similar) applications amounts to mapping the abstract annotations to values (the soundness of which relies on Proposition \ref{prop:commutes_with_homo}), and performing computation in the resulting structure (e.g., deletion propagation, access control, certification). We show graphs for the representative application of deletion propagation, since for this application there is also a baseline alternative that does not use provenance: applying the deletion directly to the input database, and then running the ``vanilla" transaction (this baseline is again labeled ``No provenance" in the relevant graphs). 

\subsection{Overhead and Usage}
\label{sec:overheadAndUse}

Figures \ref{fig:tpcc} and \ref{fig:synth} show the time and memory overhead of provenance generation as well as the time it takes to use provenance, for both TPC-C and our synthetic datasets resp.  For the latter we have set the number of affected tuples to be 200 ($0.02\%$
of the database tuples), which is
consistent with the observed percentage in TPC-C. Below (Section \ref{sec:affTuples}) we present results obtained when varying this percentage.

\begin{figure*}
	\centering
	\begin{subfigure}[t]{0.33\textwidth}
		\includegraphics[width = \linewidth]{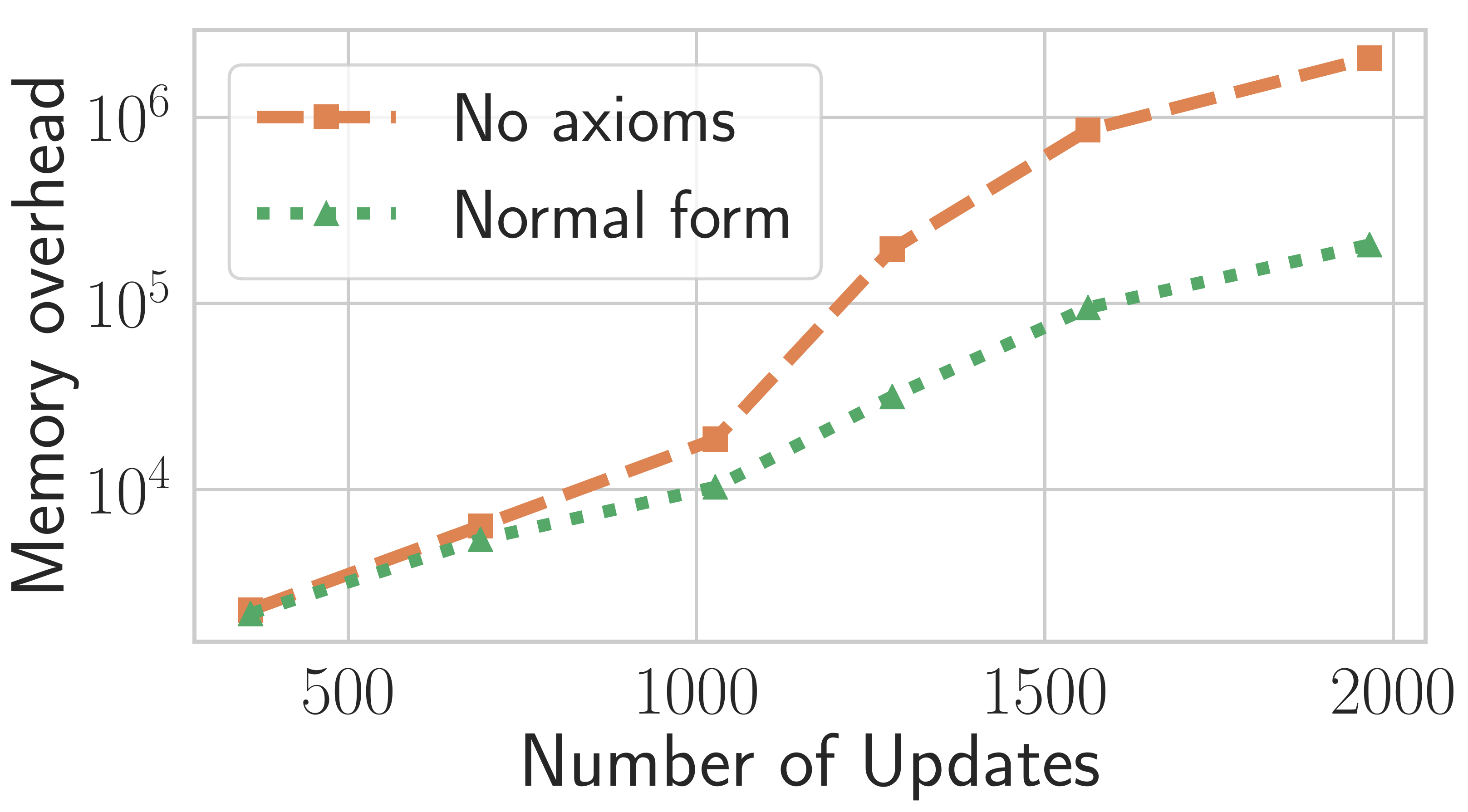}
		\caption{Memory overhead}
		\label{fig:db_size}
	\end{subfigure}%
	\begin{subfigure}[t]{0.33\textwidth}
		\includegraphics[width = \linewidth]{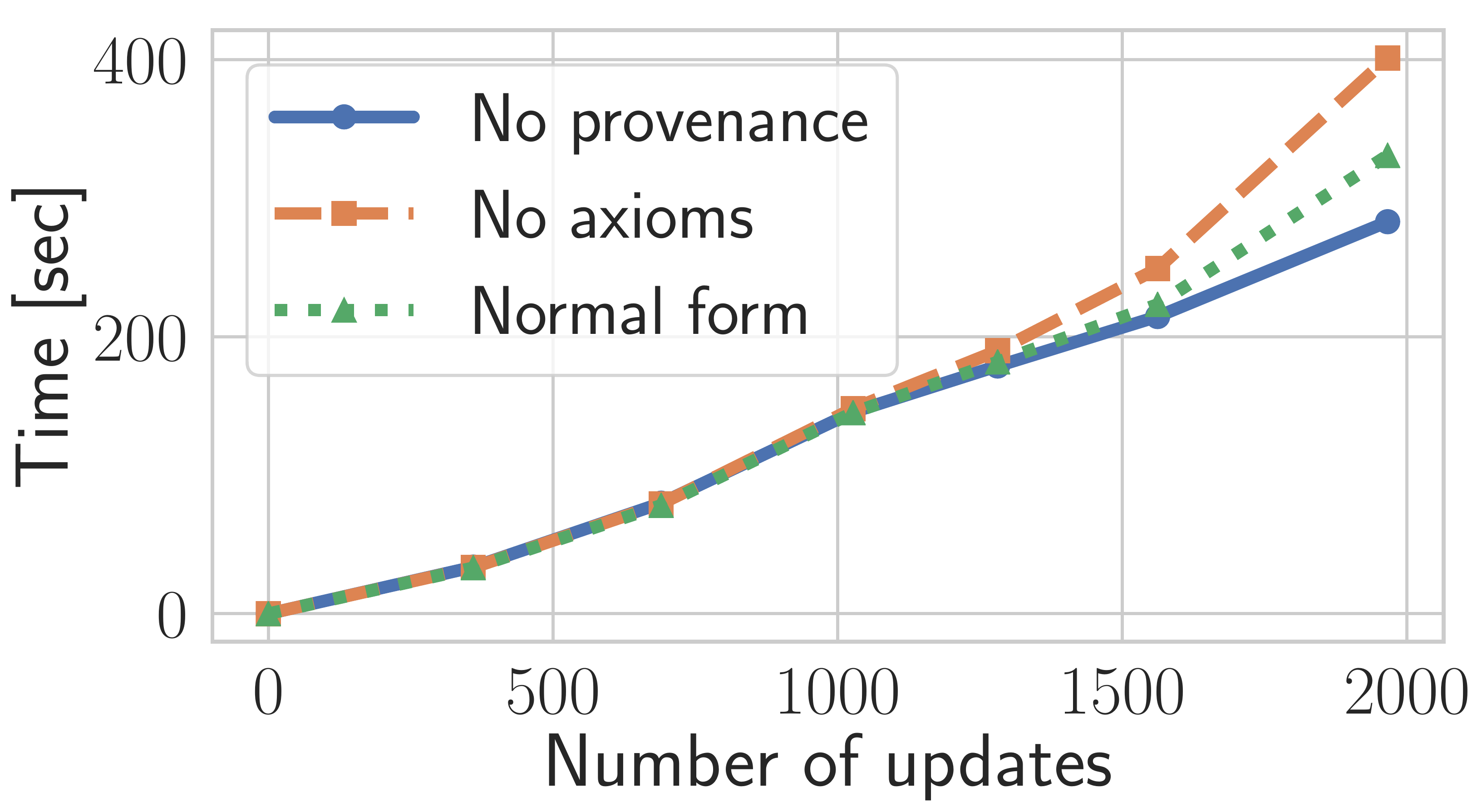}
		\caption{Runtime}
		\label{fig:runtime}
	\end{subfigure}%
	\begin{subfigure}[t]{0.33\textwidth}
		\includegraphics[width = \linewidth]{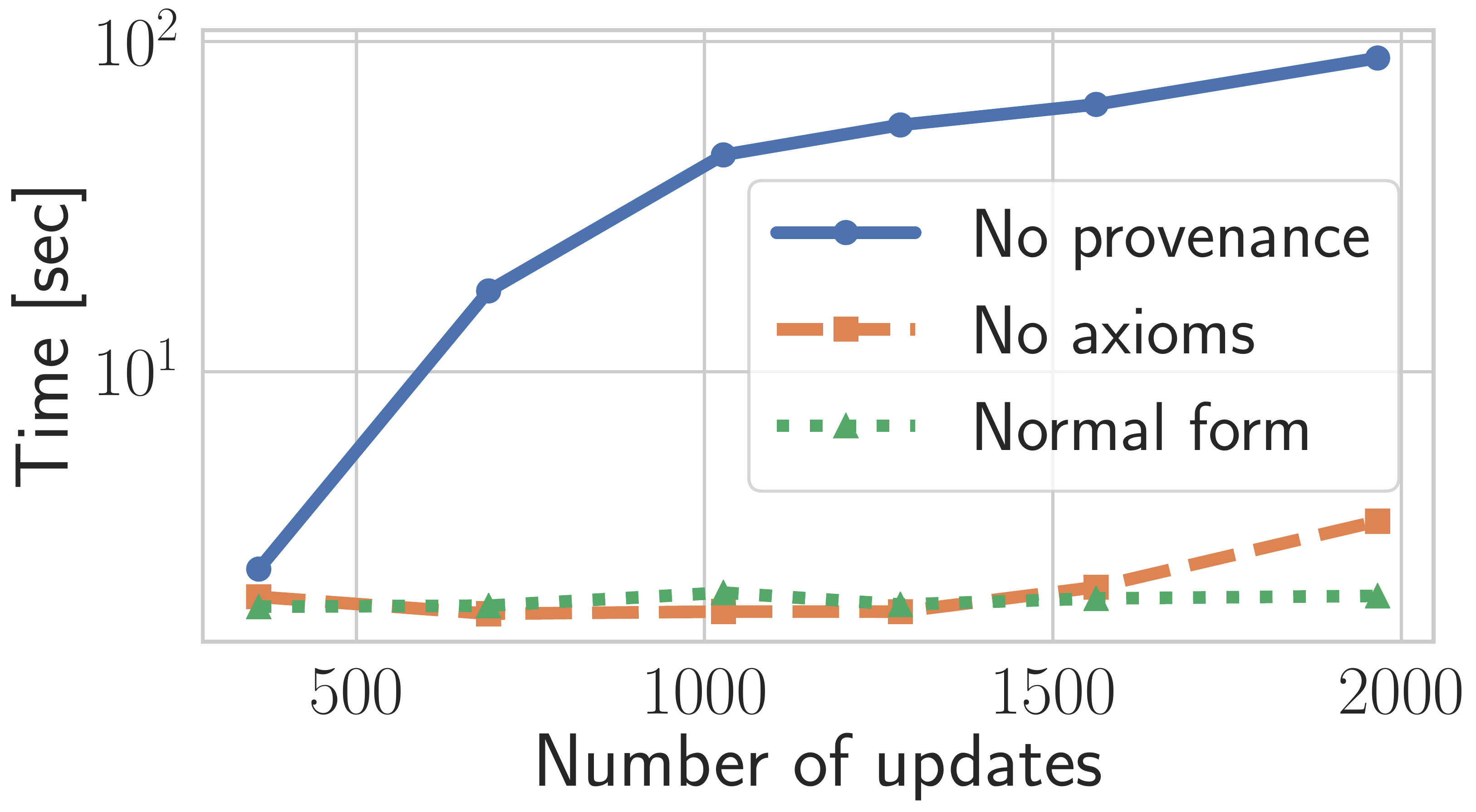}
		\caption{Usage time for deletion propagation}
		\label{fig:assignment}
	\end{subfigure}
	\caption{Provenance overhead and usage (TPC-C dataset)} \label{fig:tpcc}

\end{figure*}

\begin{figure*}

	\centering
	\begin{subfigure}[t]{0.33\textwidth}
		\includegraphics[width = \linewidth]{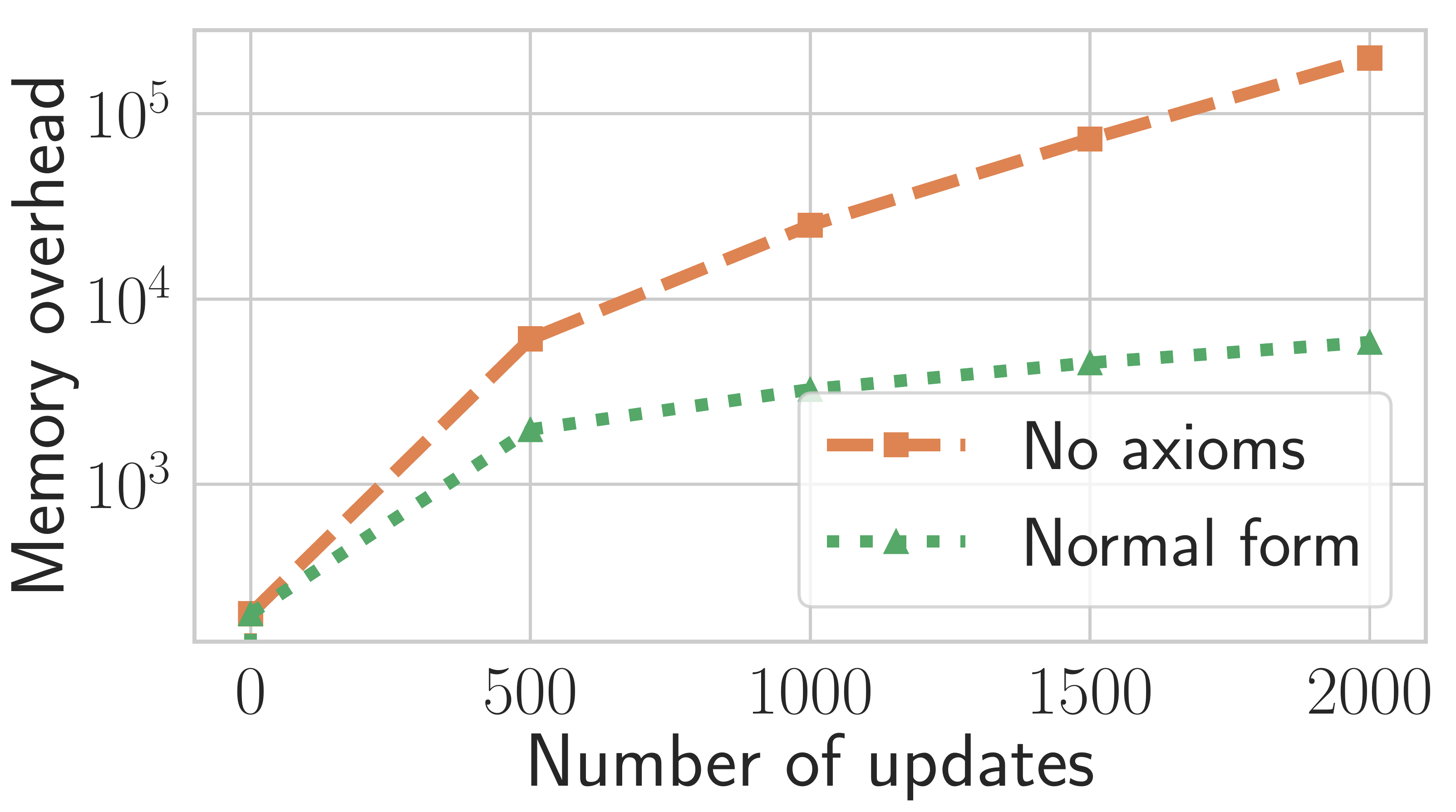}
		\caption{Memory overhead}
		\label{fig:mem_syn}
	\end{subfigure}%
	\begin{subfigure}[t]{0.33\textwidth}
		\includegraphics[width = \linewidth]{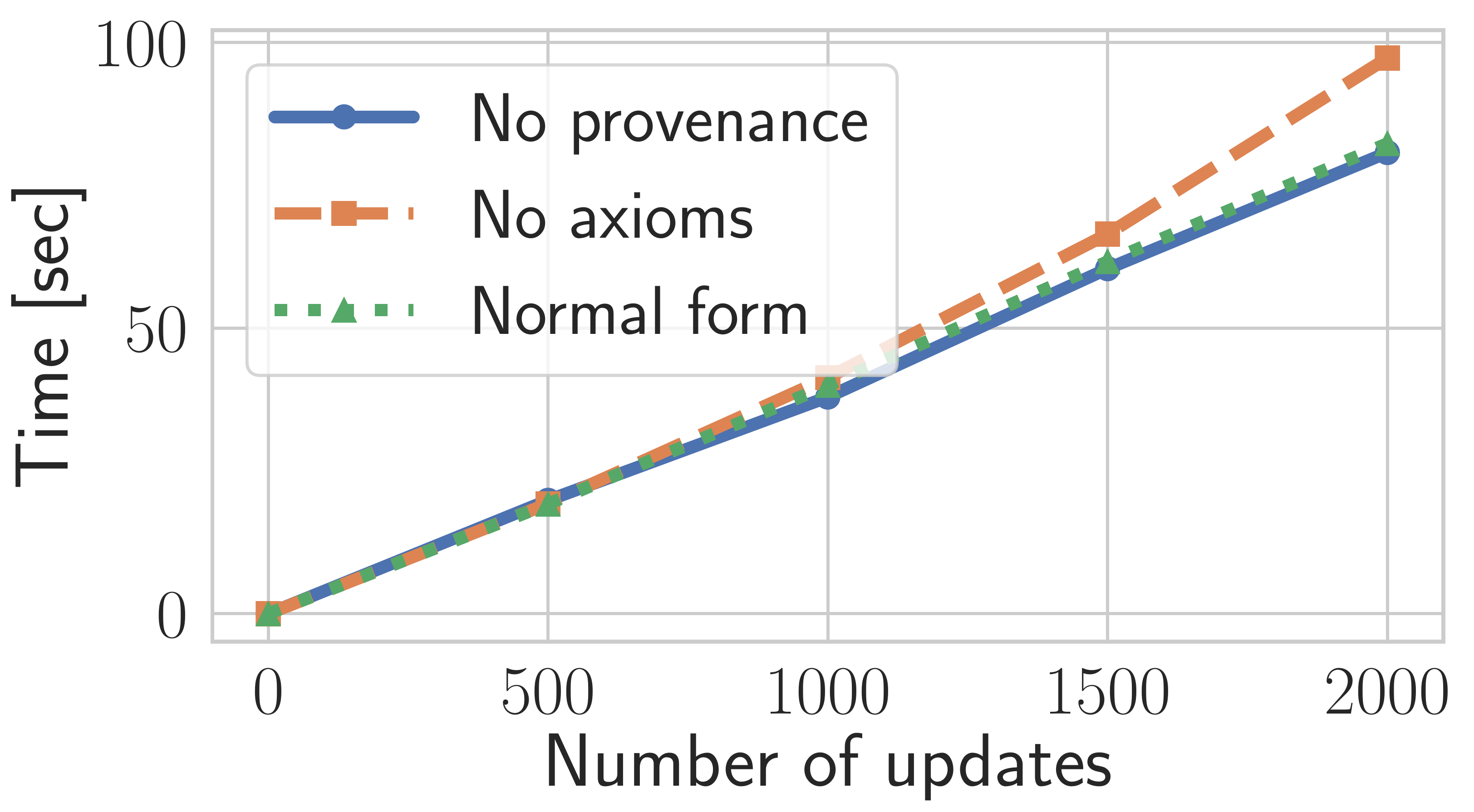}
		\caption{Runtime}
		\label{fig:runtime_syn}
	\end{subfigure}%
	\begin{subfigure}[t]{0.33\textwidth}
		\includegraphics[width = \linewidth]{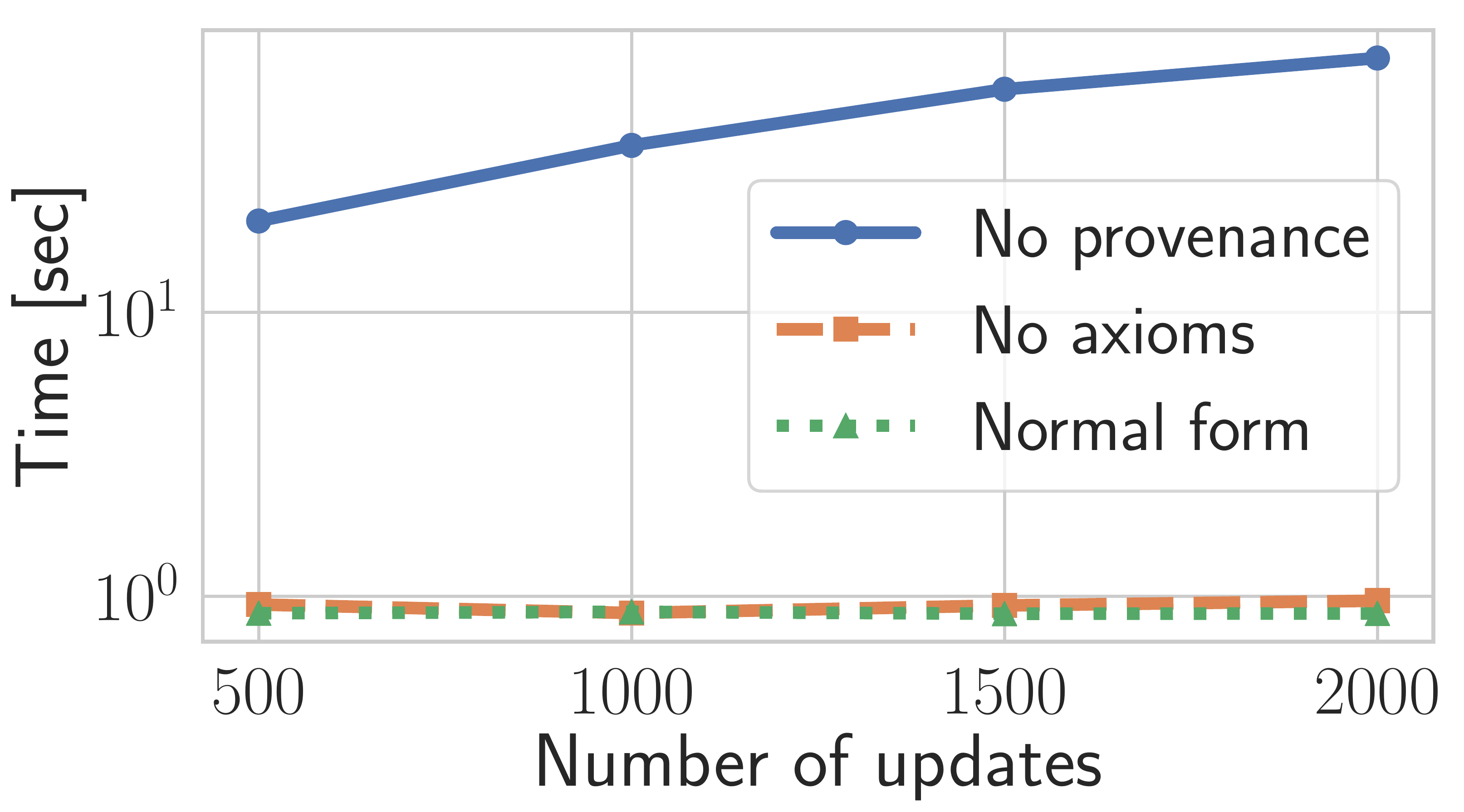}
		\caption{Usage time for deletion propagation}
		\label{fig:assignment_syn}
	\end{subfigure}

	\caption{Provenance overhead and usage (synthetic dataset)} \label{fig:synth}

\end{figure*}

\paragraph{Memory overhead} Provenance tracking leads to memory overhead of two flavors. First, recall that deleted and modified tuples
are in fact not removed from the database in our construction
(intuitively so that the operation may be ``undone"). Therefore, the
database size continuously grows. Second, maintaining the provenance
expressions incurs an overhead. 
 Figures
\ref{fig:db_size} and \ref{fig:mem_syn} show the memory overhead incurred by our construction with and without the normal form, 
compared to executing the transactions with no provenance tracking,
as a function of the number of updates. 

We note that the choice of provenance representation does
not affect the number of tuples in the database: provenance tracking
with or without the normal form representation leads to the same
number of tuples. The overhead in the database size was about
$2\%$ compared to no provenance tracking for both. In contrast to the 
 database size, there is a significant difference in the
provenance size: for the largest number of updates, the provenance
size using the naive approach (i.e., no application of axioms) was
4,127,127, while using the normal form representation the size of
the provenance was only 2,264,798, a difference of over $82\%$.

Using the synthetic dataset with $1$M tuples, we observed an overhead of about $100\%$ (i.e., $\times2$)
using the normal form representation, while the overhead without applying the
axioms was  $120\%$ with respect to no
provenance tracking.

\paragraph{Running time} Figure \ref{fig:runtime} depicts the running time of the transaction for the TPC-C dataset.
Although provenance tracking and maintenance incur overhead in both
the database size and additional memory for the provenance
information, the overhead is reasonable: the running time without
provenance tracking was $283$ seconds for the largest number of
updates, and $401$ and $330$ seconds for the provenance tracking
without using the axioms and with the normal form representation
respectively. 
When the number of updates per
tuple is small the overhead of maintaining the provenance is
negligible compared to no provenance evaluation, moreover, there is
no overhead of processing the axioms. As this number increases (after around 1K updates), the
provenance overhead increases, and the affect of the axioms is more
noticeable. Yet, the overhead of processing the rules compared to no
provenance tracking increases as well. 

Interestingly, even though using the normal form
requires the application of rules for minimization (see Section
\ref{sec:provmin}), the running time of the construction with normal
form representation is lower than that of the naive approach. This
is because the minimization is done incrementally after each update,
and as a result the maintained provenance size is significantly
smaller than the provenance expression obtained without using the
axioms. Note that generating new provenance expression for new or
updated tuples uses the existing tuples provenance and requires
copying it. Thus large provenance expressions lead also to overhead
in the tracking time, which underlines another useful aspect of the
normal form.

We observed similar trends for the synthetic dataset as shown in
Figure \ref{fig:runtime_syn}. The computation time of the
transaction with no provenance tracking was about $77$ seconds;
provenance tracking without using the axioms incurred an overhead of
over $25\%$ (about $97$ seconds), whereas using the normal form
representation, the running time overhead was less than $3\%$ (only
$79$ seconds).

\vspace{-0.3cm}

\paragraph{Provenance Usage}  As explained above, we have examined the time it takes to use provenance for deletion propagation (with and without the normal form), compared to a baseline that re-computes the transaction result for the deletion scenario. The results are reported in Figures
\ref{fig:assignment} and \ref{fig:assignment_syn}.


For the two datasets and for both variants of provenance tracking,
using the provenance framework significantly outperforms the
baseline approach. For the largest number of updates in the TPC-C
dataset (Figure \ref{fig:assignment}), re-running the transaction
over the modified database took $89$ seconds, while the provenance
assignment time was $3.43$ seconds (over $\times 25$ faster) for the naive construction and $1.94$ seconds using the
normal form representation (over $\times 45$ faster that the
baseline). The gain of using the normal form representation compared
with the naive construction was significant: about $78\%$. For the
synthetic dataset (Figure \ref{fig:assignment_syn}), the
re-computation time was $78$ seconds, the assignment time for the
provenance generated without using the axioms was $0.96$ seconds,
and for the normal form it was $0.86$ seconds. These are over
$\times 81$ and $\times 91$ faster than the baseline, respectively.

\subsection{When do we gain from the Normal Form?}\label{sec:affTuples}
The next set of experiments aims at assessing the usefulness of the
normal form representation in synthetic environment where we change
the provenance size.  As the number of update per tuple increases,
the difference between the sizes of the provenance generated without
using the axioms and of the provenance represented in the normal
form, increases. For a fixed number of updates, as the number of the overall
affected tuples increases, the number of update per tuple decreases
(since the updated tuples are selected with uniform distribution).
Thus, we fixed the transaction length and examined the effect of the
number of tuples affected by the transaction on the overhead
incurred by provenance tracking with and without the normal form. We
varied the number of affected tuples from $200$ to $1000$ ($0.1\%$
of the database size). This is in line with the number of affected
tuples in the TPC-C dataset that varies from from $200$ to $2000$
($0.1\%$ of the database size there). The results for $1$M tuples
and $2000$ update queries are shown in Figure
\ref{fig:query_part}.

The right-hand side of Figure \ref{fig:query_part} depicts the memory overhead of
provenance tracking as a function of the overall number of  tuples affected
by the transaction. Recall that the axioms allow us to compactly
represent the provenance expression of a single tuple at a time.
Thus, for large number of updates per tuple, we expect to see a
significant difference between the two approaches. Indeed, for small
numbers of affected tuples, the provenance size of each tuple is
larger. Then, the effect of the axioms on the provenance size is
notable, reflecting on the memory overhead. We note that there is a
moderate growth in the memory overhead when using the axioms as the
number of affected tuples increases.


\begin{figure}
	\centering
	\begin{subfigure}[b]{\textwidth}
		\includegraphics[width = 0.45\linewidth]{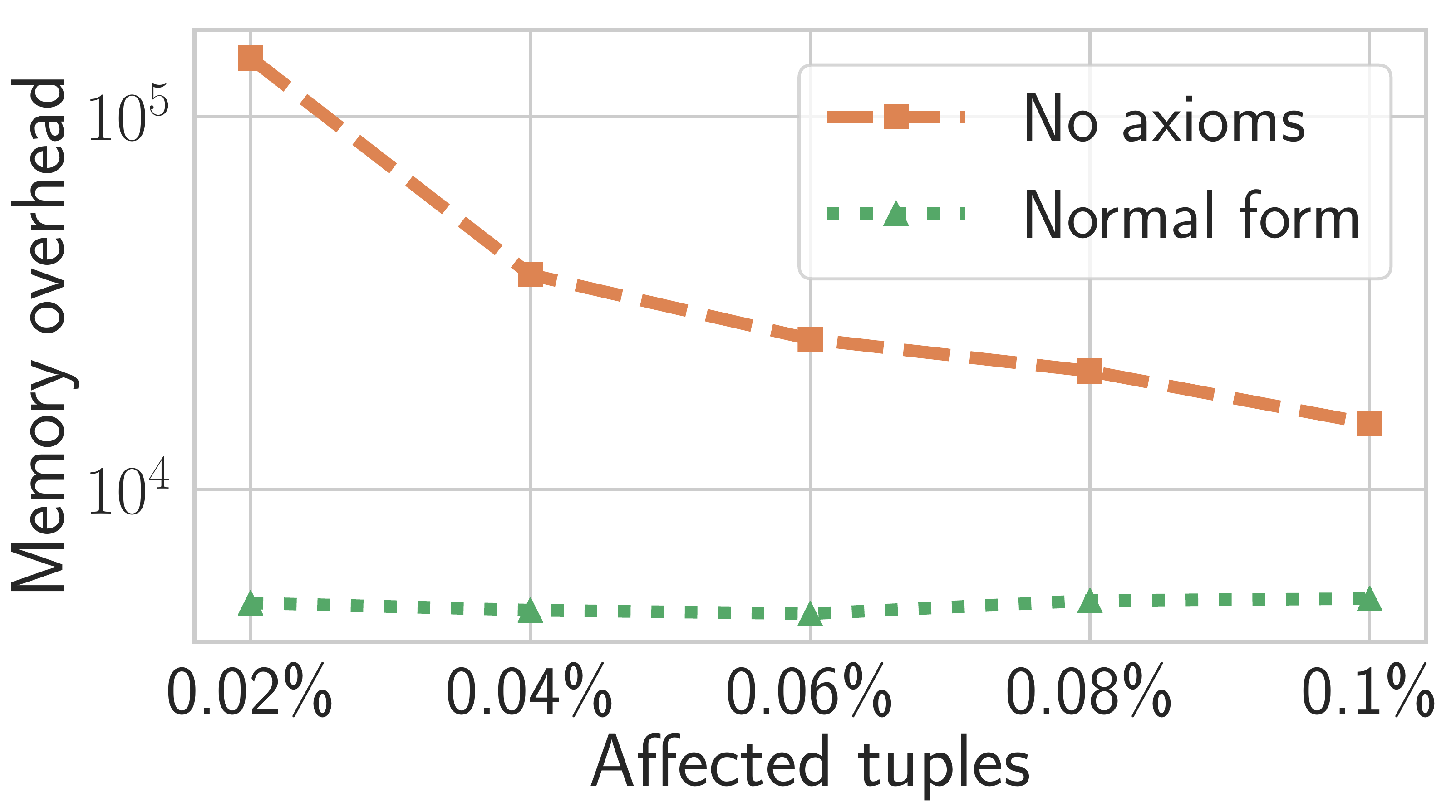}
		~
		\includegraphics[width = 0.47\linewidth]{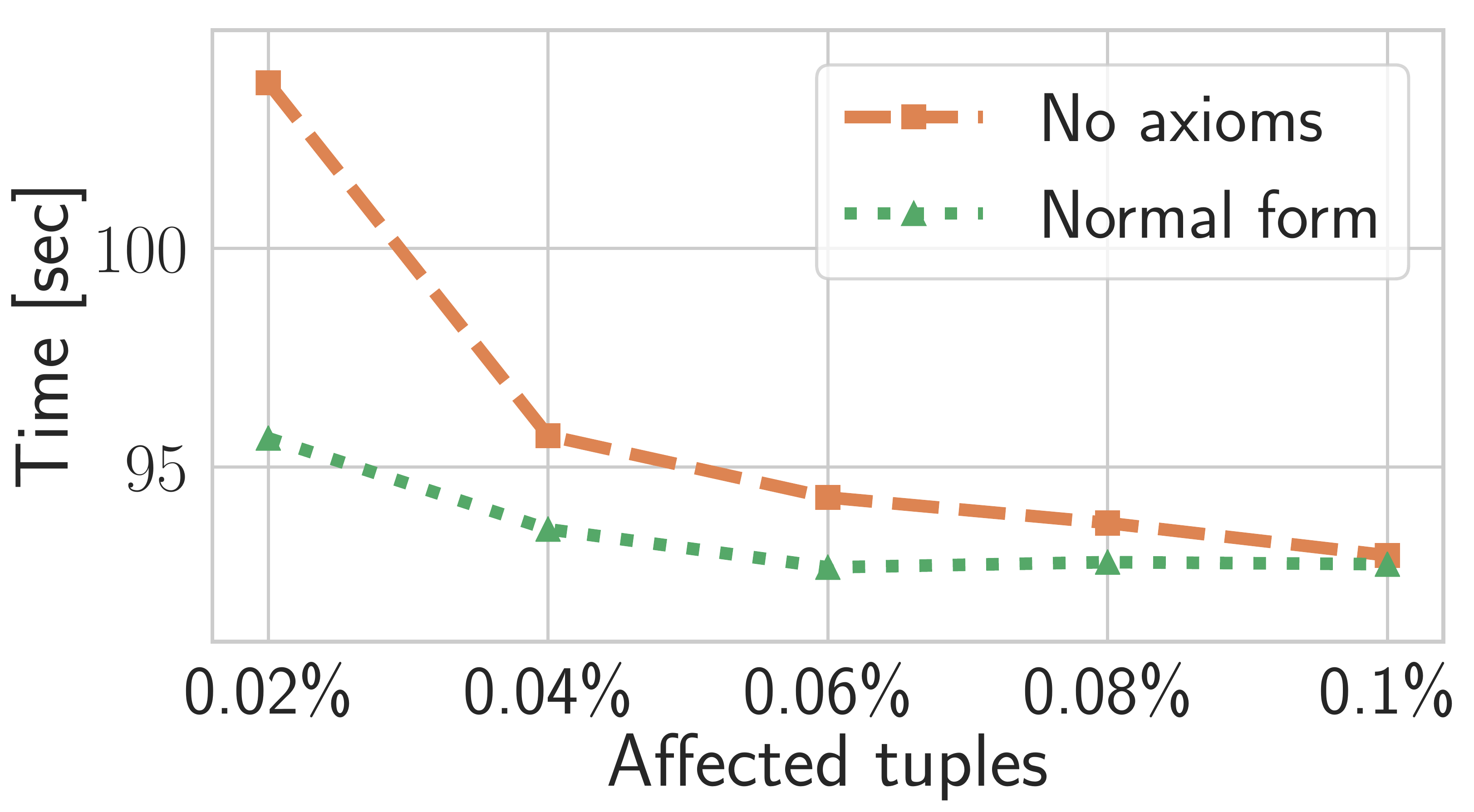}
		\caption{Total number of overall affected tuples}
		\label{fig:query_part}
	\end{subfigure}%
	
	\begin{subfigure}[b]{\textwidth}
		\includegraphics[width = 0.45\linewidth]{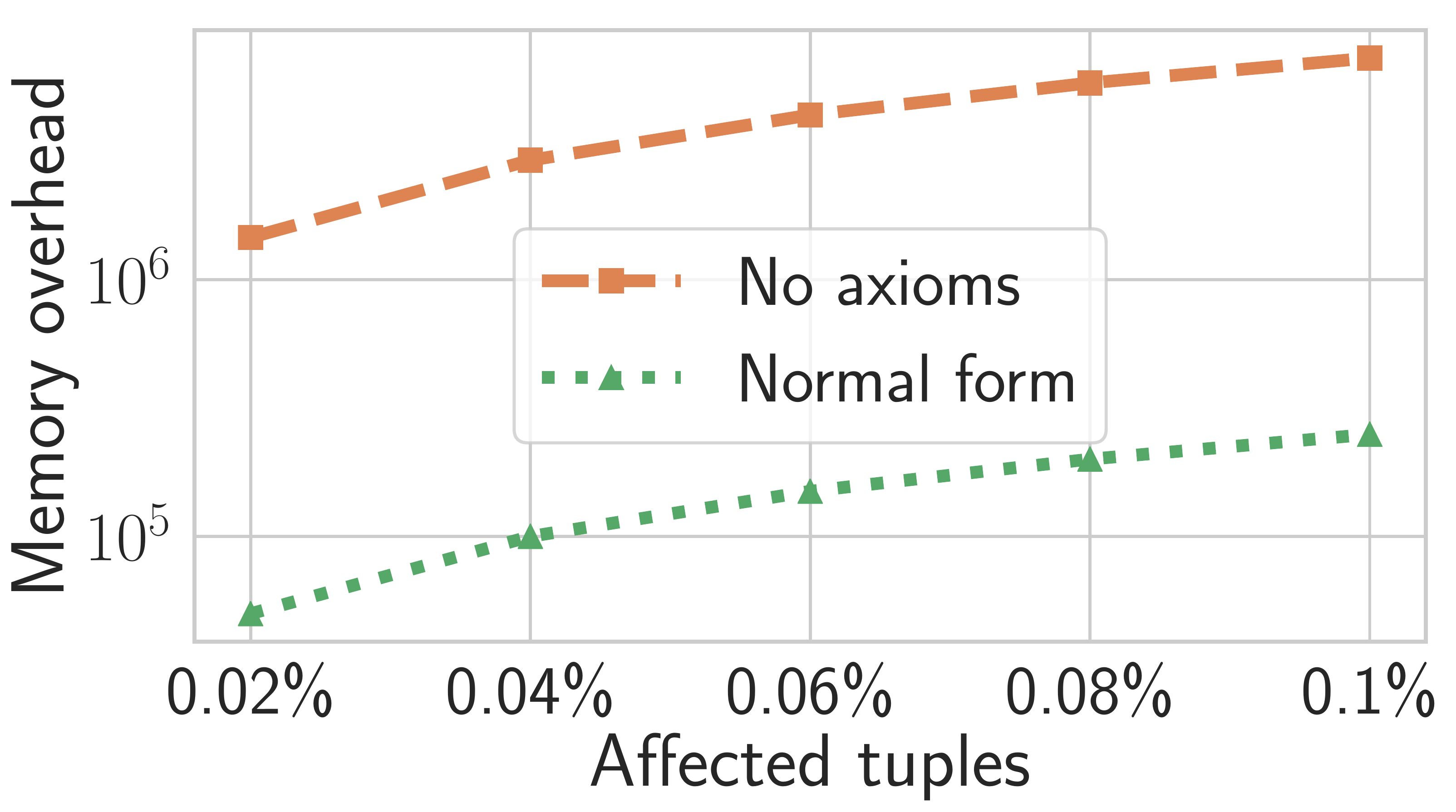}
		~
		\includegraphics[width = 0.47\linewidth]{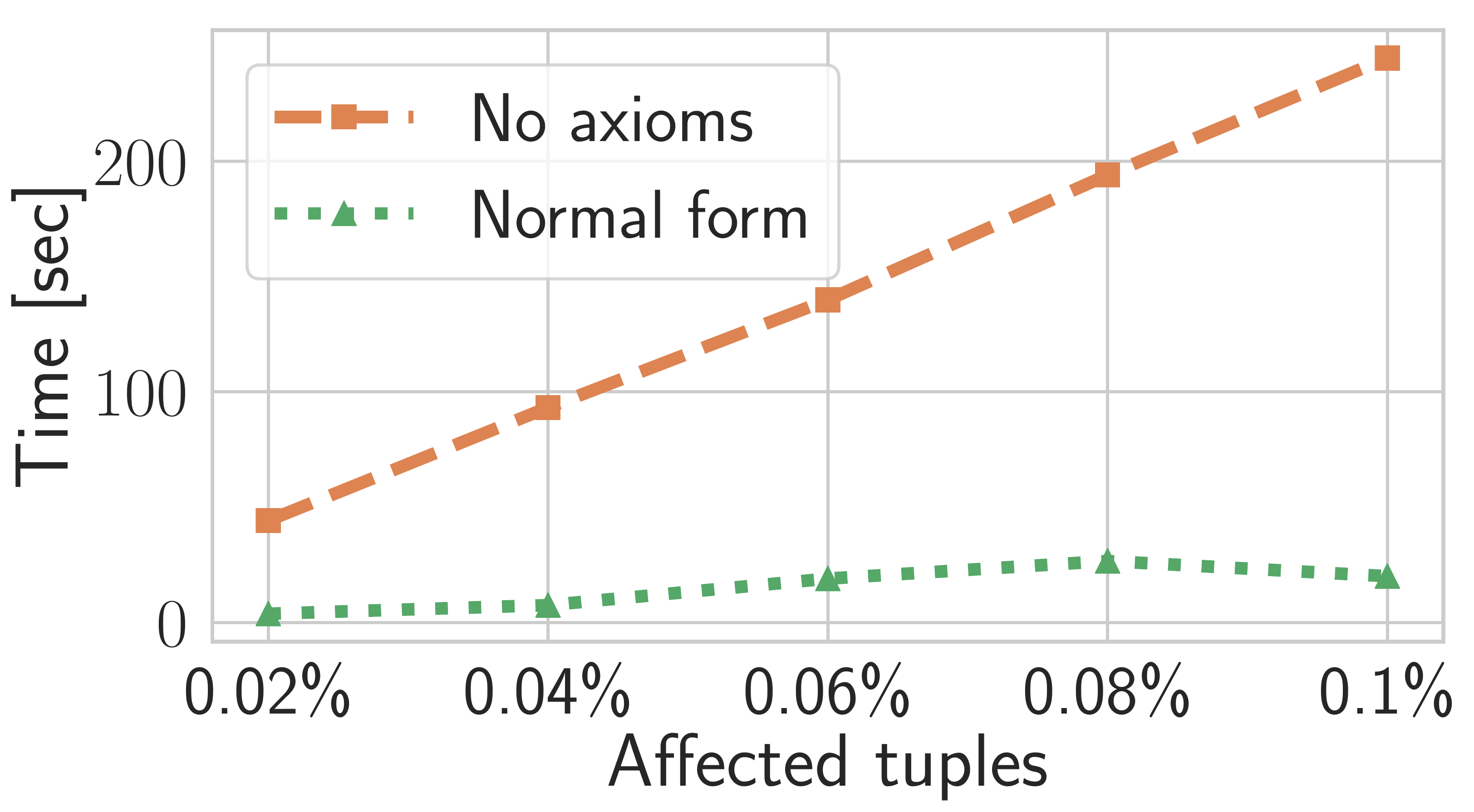}
		\caption{Number of affected tuples per query}
		\label{fig:query_part_cr}
	\end{subfigure}%

	\caption{Naive representation Vs. normal form as a function of number affected tuples} \label{fig:syn_query_par}

\end{figure}


The running time as a function of the total number of affected tuples is presented in the right-hand side of Figure \ref{fig:query_part}. As a
result of the changes in the provenance size when the number of
updates per tuple increases, the overhead of maintaining the
provenance without using the axioms increases. We also observed a
moderate growth for the construction that uses the axioms when the
number of update per tuple increases. This growth is due to the
(relatively small) overhead of minimizing the provenance after each
update.

To highlight the difference between the provenance tracking approaches we examined the effect of the number of tuples affected by each update query. To this end, we fixed the data size, and the transaction length, and increase the number of tuples affected by each update. Figure \ref{fig:query_part_cr} shows the results for a database with $1$M tuples and $5$ update queries. We observed a moderate growth in the memory overhead (left-hand side of the figure), for both methods, with a significant lower overhead using the axioms. There is a notable difference in the running time growth (right-hand side of the figure), as a result of the large overhead incur by managing large provenance expressions without using the axioms.

\subsection{Comparison with MV-semirings \cite{glavic}}

\begin{figure}
	\centering
	\begin{subfigure}[t]{0.5\linewidth}
		\includegraphics[width = \linewidth]{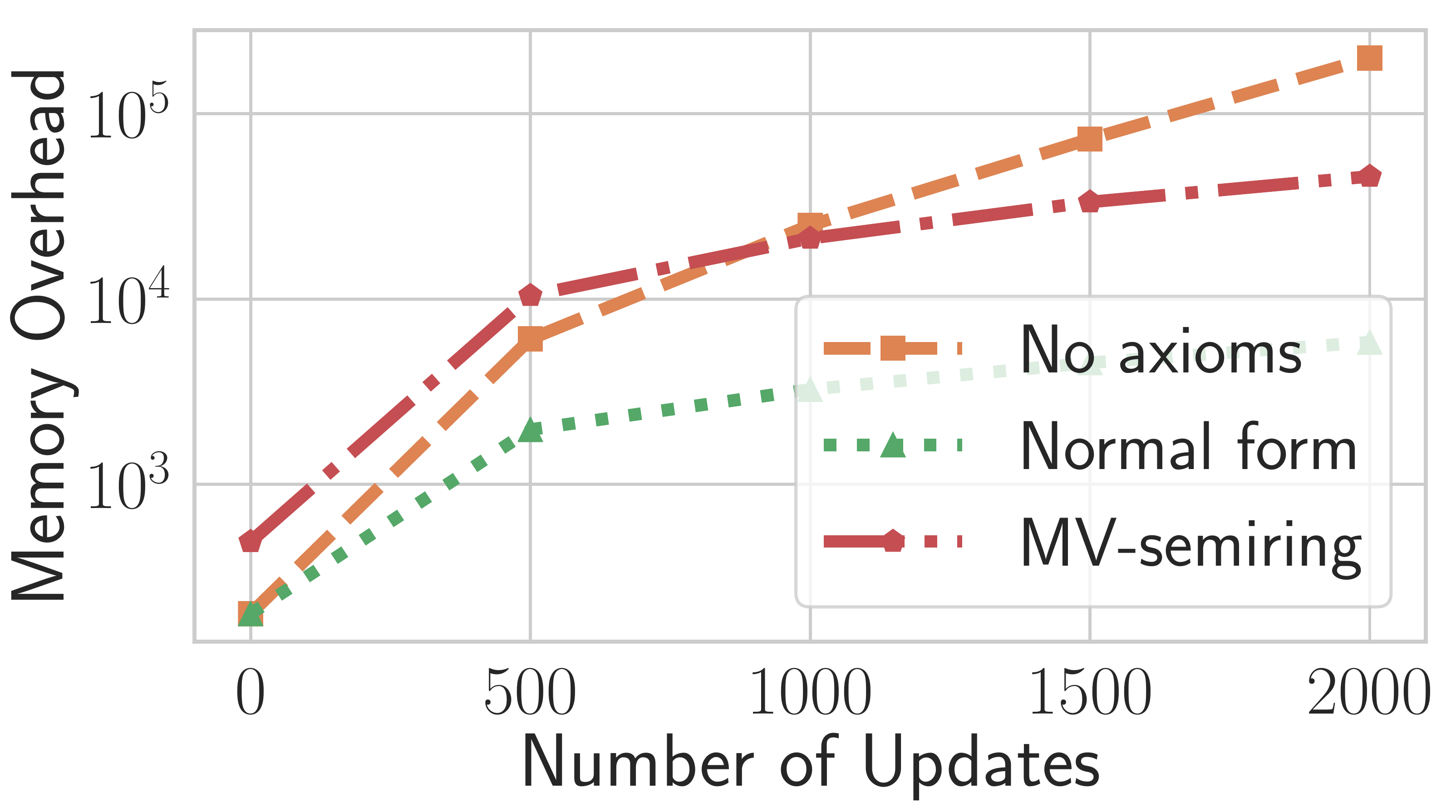}
		\caption{Memory Overhead}
		\label{fig:exp_len}
	\end{subfigure}%
	\begin{subfigure}[t]{0.5\linewidth}
		\includegraphics[width = \linewidth]{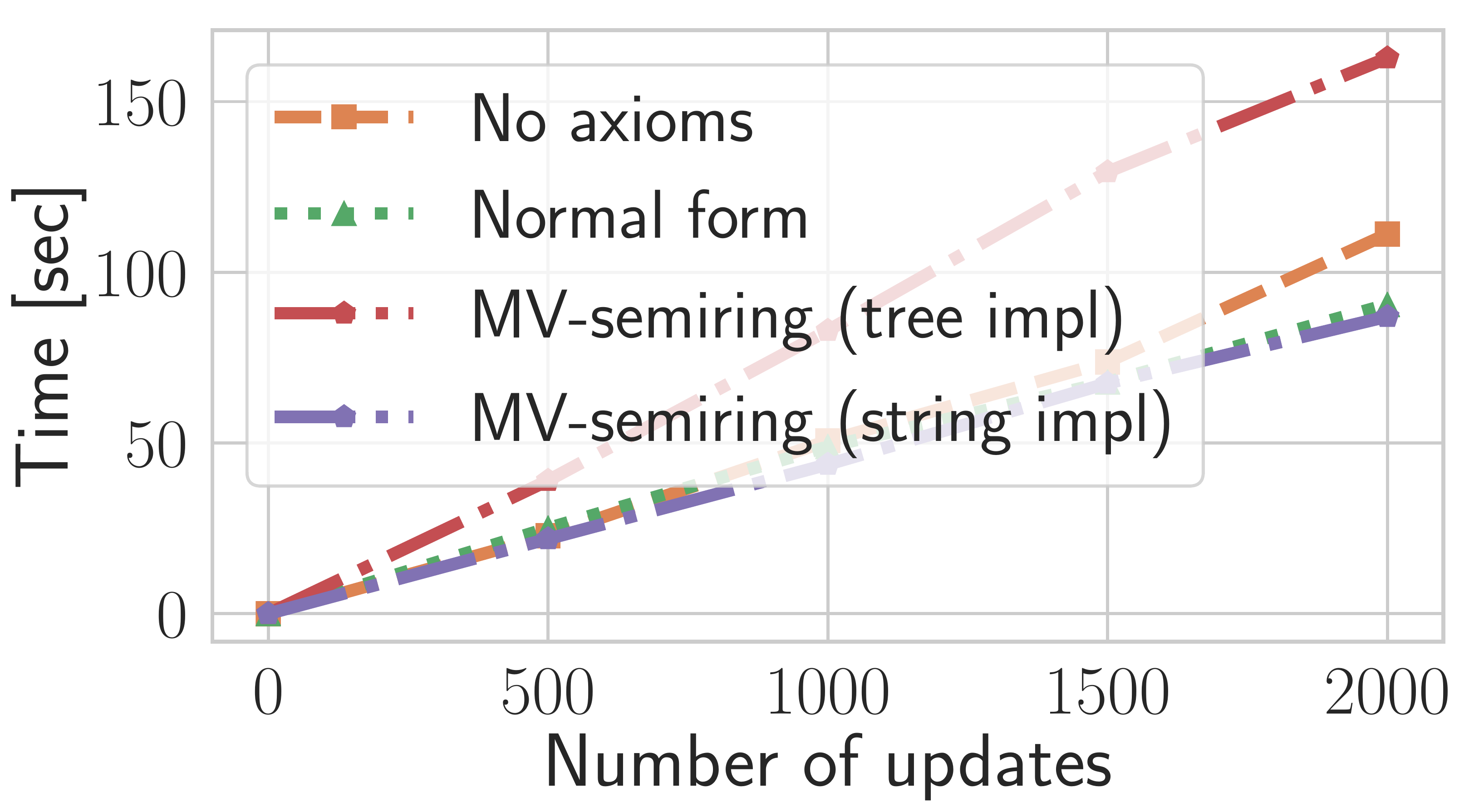}
		\caption{Runtime}
		\label{fig:run_time_glavic}
	\end{subfigure}

	\caption{Comparison with \cite{glavic} (synthetic dataset)} \label{fig:glavic}

\end{figure}

\label{sec:compWithGlavic}
We conclude with an experimental comparison to the MV-semiring model proposed in \cite{glavic}. We have implemented a generator of MV-semiring expressions and used it to compare to our solution. We note that the model of \cite{glavic} is geared towards different use cases than ours and stores somewhat different information. In turn, the intended use case could have significant effect on the implementation (e.g., choice of data structures to represent provenance) and in turn on the algorithms performance. Another difference is in that, as explained above, for our applications we need to ``duplicate" modified tuples, while \cite{glavic} does not.

To this end, in order to get an implementation-independent assessment of the memory consumption, we measure the sum of the total provenance length and the number of database tuples. Figure \ref{fig:exp_len} shows the memory overhead for both approaches  compared with no provenance tracking evaluation. While the provenance length of individual tuples using the model of \cite{glavic} is roughly the same as that of our model without using the axioms, the number of tuples in the resulting database using our model is larger, and thus the memory overhead of our model with no axioms is higher. However when using the axioms, we obtain much smaller expressions than in the MV-semiring model.  


Figure \ref{fig:run_time_glavic} depicts the running time as the function of number of updates. Here again, performance highly depends on implementation details and we demonstrate this using our two different implementations of \cite{glavic}. The first uses strings to represent the provenance (purple line). The running time using this implementation was slightly better than our model, however it has an ``edge": it requires a parsing the provenance as pre-process for each use. The second implementation is tree based (red line), using the anytree python package, which is more similar to the implementation of our model. Our model outperforms the MV-semiring model using this implementation. This is because the trees obtained for the MV-semiring model are deep, and the large overhead for each update is incurred by their recursive structure. We estimate that most reasonable implementations would likely to perform in the range between our two implementations, depending on the intended use.

%% file: related.tex
\section{Related Work}\label{sec:related}

Data provenance has been studied
for different data transformation languages, from relational algebra
to Nested Relational Calculus, with different provenance models
\cite{Trio,GlavicSAM13,GeertsP10,KenigGS13,CheneyProvenance,w3c,ProvenanceBuneman,Olteanu12})
and applications \cite{Suciu,Meliou2,RoyS14,Gatt,GlavicAMH10}. Our work fits the line of research on algebraic provenance, originating in \cite{GKT-pods07} for positive relational algebra. Consequent
algebraic constructions have since been
proposed for various formalisms including aggregate queries
\cite{pods11}, queries with difference \cite{tapp11}, Datalog
\cite{GKT-pods07} and SPARQL queries \cite{GeertsUKFC16}. 

A provenance model for SPARQL queries and updates using a provenance graph was presented in \cite{CheneyUpdates}, and \cite{ProvenanceBuneman} proposes an approach to
provenance tracking for data that is copied among databases using a sequence of insert,
delete, copy, and paste actions. Provenance for updates was also studied in \cite{VC,BunemanCC06} and in \cite{icde16}, where a boolean provenance model for updates was proposed; however, none of these approaches has proposed an algebraic provenance model. Updates are a form of non-monotone reasoning, and as such are related to the notion of relational difference. Algebraic provenance models for queries with difference have been proposed in \cite{tapp11,GeertsP10,GeertsUKFC16}, but naturally none could be directly applied to update queries; in particular using the ``monus" operation of \cite{GeertsP10,GeertsUKFC16} as our minus operation may not result in a structure satisfying the equivalence axioms.  Further exploration of the connections between these models and ours is left for future work. An extension of the semiring model of \cite{GKT-pods07} to account for updates was also studied in \cite{KarvounarakisGIT13}, however the focus there is on the use of provenance in the context of trust and while the work includes an efficient implementation, it falls short of proposing a generic algebraic construction. Closest to our work is the multi-version semiring  (MV-semiring) model \cite{glavic} to which we have extensively compared our solution. 

We have shown a normal form construction that allows significant
reduction of provenance size in the context of hyperplane update
queries. Provenance size reduction has been studied in multiple
additional contexts. In particular, the work of \cite{jag} has shown
a highly effective method for summarizing {\em workflow provenance},
namely the workflow operations (modules) applied to a data item in
the course of execution. The provenance model used in \cite{jag} is
geared towards workflows. It thus captures module invocations, their
arguments etc., which are absent from our model. On the other hand,
it thus does not capture fine-grained combinations of data items
that are captured in algebraic models such as the one we present
here. Consequently, their method (which includes, e.g., argument
factorization) is not applicable to our setting (nor does our method
relevant for their setting). Similar considerations distinguish our
work from \cite{boon}, that studies compression of {\em network
provenance}. Such provenance includes information that is different
from ours, involving a record of network events (albeit using a data
provenance representation through a datalog-like formulation of the
network logic), rather than information on data derivation in
general and data updates in particular. 

In contrast, the work of \cite{olteanu,olteanutapp} does focus on
algebraic provenance expressions. The expressions studied there are
provenance polynomials in the sense of \cite{GKT-pods07} (elements
of the $N[X]$ semiring), which are designed to capture provenance
for SPJU queries but do not suffice for update queries; in
particular no counterpart of a minus operator is present there.
Consequently, the factorization methods in
\cite{olteanu,olteanutapp} are very different from our normal form
construction, in particular because they rely on a different set of
operators and axioms. Our additional operators and different axioms
entail that the methods of \cite{olteanu,olteanutapp} are
inapplicable for provenance compression for update queries (we note
that \cite{olteanu} also studies compression of query results, not
only provenance; this is orthogonal to our work); on the other hand
our construction depends on limiting the queries to hyperplane
queries, which means that our solution is also not applicable to
general SPJU (or even SPJ) queries.

%% file: conc.tex
\section{Conclusion and Limitations}
\label{sec:conc}

We have developed a novel algebraic provenance model for hyperplane
update queries and sequences thereof, following the axiomatization
of query equivalence in \cite{vianu}. We have shown that the model
captures the ``essence of computation" for such queries, i.e.,
equivalent transactions yield equivalent provenance. We have shown
means of instantiating the model, towards applications of provenance
in this context. The example instances show the usefulness of the
generic model: by following the axioms, we are guaranteed that our
provenance construction is independent of transaction rewrites. We
have further studied the efficient computation and storage of
provenance, and have shown a minimization technique that leads to
compact provenance representation via a normal form. This again
leverages the axioms, this time in a computational manner. Our
experimental evaluation shows the tractability and usefulness of the
approach, as well as the benefit of using the normal form.

A main limitation of our solution is that it is confined to
hyperplane queries. One may address this challenge towards
supporting update queries with conjunctive conditions and beyond;
yet such attempt would likely cost in the loss of the property of
provenance being preserved under equivalence, since (to our
knowledge) no sound and complete axiomatization is known for these
more expressive fragments. Further exploration of such endeavours is
left for future work.